\documentclass[11pt,english]{article}
\pdfoutput=1
\usepackage[T1]{fontenc}
\usepackage[latin9]{inputenc}
\usepackage{verbatim}
\usepackage{amsthm}
\usepackage{amsmath}
\usepackage{amssymb}
\usepackage{graphicx}
\usepackage{enumerate}
\PassOptionsToPackage{normalem}{ulem}
\usepackage{ulem}

\makeatletter

\usepackage[dvipsnames]{xcolor}
\newcounter{note}[section] 

\theoremstyle{plain}
\newtheorem{thm}{\protect\theoremname}
  \theoremstyle{plain}
  \newtheorem{prop}[thm]{\protect\propositionname}
  \theoremstyle{plain}
  \newtheorem{lem}[thm]{\protect\lemmaname}
  \theoremstyle{definition}
  \newtheorem{defn}[thm]{\protect\definitionname}
    \theoremstyle{plain}
  \newtheorem{cor}[thm]{\protect\corollaryname}

\usepackage{fullpage}

\makeatother

\usepackage{babel}
  \providecommand{\corollaryname}{Corollary}
  \providecommand{\definitionname}{Definition}
  \providecommand{\lemmaname}{Lemma}
  \providecommand{\propositionname}{Proposition}
\providecommand{\theoremname}{Theorem}

\begin{document}
\date{}
\thispagestyle{empty}

\global\long\def\A{\mathcal{A}}
 \global\long\def\B{\mathcal{B}}
 \global\long\def\L{\mathcal{L}}
 \global\long\def\Lb{\bar{\L}}
 \global\long\def\R{\mathcal{R}}
 \global\long\def\polylog{\mathrm{polylog\,}}
 \global\long\def\P{\mathcal{P}}
 \global\long\def\C{\mathcal{C}}
 \global\long\def\int{\mathrm{int}}
 \global\long\def\Rs{\R_{\mathrm{save}}}
 \global\long\def\Q{\mathcal{Q}}
\global\long\def\Rb{\bar{\R}}
\global\long\def\F{\mathcal{F}}
\global\long\def\Le{\L_{\mathrm{ext}}}
\global\long\def\Lg{\L_{\mathrm{grid}}}
\global\long\def\Lc{\L_{\mathrm{cut}}}

\global\long\def\eps{\varepsilon}

\title{Approximation Schemes for Maximum Weight \\Independent Set of Rectangles}
\author{
Anna Adamaszek\footnote{Max-Planck-Institut f\"ur Informatik, Saarbr\"ucken, Germany,
       \texttt{\{anna,awiese\}@mpi-inf.mpg.de}}
\and
Andreas Wiese\footnotemark[1]
}

\maketitle
\begin{abstract}
In the Maximum Weight Independent Set of Rectangles (MWISR) problem
we are given a set of $n$ axis-parallel rectangles in the 2D-plane,
and the goal is to select a maximum weight subset of pairwise non-overlapping
rectangles. Due to many applications, e.g. in data mining, map labeling
and admission control, the problem has received a lot of attention
by various research communities. We present the first $(1+\eps)$-approximation
algorithm for the MWISR problem with quasi-polynomial running time
$2^{\mathrm{poly}(\log n/\eps)}$. In contrast, the best known polynomial time
approximation algorithms for the problem achieve superconstant approximation
ratios of $O(\log\log n)$ (unweighted case) and $O(\log n/\log\log n)$
(weighted case). 

Key to our results is a new geometric dynamic program which recursively
subdivides the plane into polygons of bounded complexity. We provide
the technical tools that are needed to analyze its performance. In particular, we present a 
method of partitioning the plane into small and simple areas such that the rectangles of an optimal
solution are intersected in a very controlled manner. Together with a novel application
of the weighted planar graph separator theorem due to Arora et al.~\cite{Arora1998}
this allows us to upper bound our approximation ratio by $1+\eps$. 

Our
dynamic program is very general and we believe that it will be useful
for other settings. In particular, we show that, when parametrized
properly, it provides a \emph{polynomial time }$(1+\eps)$-approximation
for the special case of the MWISR problem when each rectangle is relatively large in at least one dimension.
Key to this analysis is a method to tile
the plane in order to approximately describe the topology of these rectangles
in an optimal solution.
This technique might be a useful
insight to design better polynomial time approximation algorithms or even 
a PTAS for the MWISR problem. In particular, note that our
results imply that the MWISR problem is not $\mathsf{APX}$-hard, unless
$\mathsf{NP}\subseteq\mathsf{DTIME}(2^{\polylog(n)})$. 
\end{abstract}
\newpage{}

\setcounter{page}{1}


\section{Introduction}

One of the most fundamental problems in combinatorial optimization
is the \textsc{Independent Set} problem: given an undirected graph,
find a set of pairwise non-adjacent vertices with maximum total weight.
While the general problem is essentially intractable (it is $\mathsf{NP}$-hard
to approximate with a factor $n^{1-\eps}$ for any $\eps>0$~\cite{zuck07}),
many special cases allow much better approximation ratios.

One extensively studied setting are graphs which stem from geometric
shapes in the 2D-plane. Given a set of geometric objects in the plane,
the goal is to find a set of pairwise non-overlapping objects with
maximum total weight. Depending on the complexity of these shapes,
the approximation factors of the best known polynomial time algorithms
range from $1+\eps$ for fat objects~\cite{EJS2001}, to $n^{\eps}$
for arbitrary shapes~\cite{FoxPach2011}. Observe that the latter
is still much better than the complexity lower bound of $n^{1-\eps}$
for arbitrary \textsc{Independent Set} instances.

Interestingly, there is a very large gap between the best known approximation
factors when the considered objects are squares of arbitrary sizes
and when they are rectangles. For squares, a $(1+\eps)$-approximation
has been known for several years~\cite{EJS2001}. For the rectangles,
the best known approximation factors are $O(\log n/\log\log n)$ for
the general case~\cite{ChanHarPeled2009}, and $O(\log\log n)$ for
the cardinality case~\cite{CC2009}. Importantly, no constant factor
approximation algorithms are known for rectangles, while the best
known hardness result is $\mathsf{NP}$-hardness~\cite{fowler1981optimal,imai1983finding}.
These gaps remain despite a lot of research on the problem~\cite{AKS1998,berman2001improved,CC2009,chan2004note,ChanHarPeled2009,fowler1981optimal,imai1983finding,KMP1998,lewin2002routing,N2000}, which is particularly motivated by its many applications in areas
such as channel admission control~\cite{lewin2002routing}, map labeling~\cite{AKS1998,doerschler1992rule},
and data mining~\cite{fukuda2001data,KMP1998,lent1997clustering}.

Since even for arbitrary shapes the best known hardness result is
$\mathsf{NP}$-hardness, it seems that more sophisticated algorithmic
techniques and/or complexity results are needed to fully understand
the \textsc{Independent Set} problem in the geometric setting.


\subsection{Related Work}

The maximum weight independent set of rectangles problem has been
widely studied. There are several $O(\log n)$ approximation algorithms
known~\cite{AKS1998,KMP1998,N2000}, and in fact the hidden constant
can be made arbitrarily small since for any $k$ there is a $\left\lceil \log_{k}n\right\rceil $-approximation
algorithm due to Berman et al.~\cite{berman2001improved}. Eventually,
a $O(\log n/\log\log n)$-approximation algorithm has been presented by Chan and Har-Peled~\cite{ChanHarPeled2009}.
Some algorithms have been studied which perform better for special
cases of~MWISR. 
There is a $4q$-approximation algorithm due to Lewin-Eytan, Naor,
and Orda~\cite{lewin2002routing} where $q$ denotes the size of
the largest clique in the given instance. In case that the optimal independent set has size
$\beta n$ for some $\beta\le1$, Agarwal and Mustafa present an algorithm
which computes an independent set of size $\Omega(\beta^{2}n)$~\cite{agarwal2006independent}.

In a break-through result, Chalermsook and Chuzhoy give a $O(\log\log n)$-approximation
algorithm for the cardinality case~\cite{CC2009}, which is based
on the natural LP-relaxation of the problem. In fact, it is a challenging
open problem to determine the exact integrality gap of the LP. Currently
the best known upper bounds for it are $O(\log n/\log\log n)$~\cite{ChanHarPeled2009}
for the weighted case, and $O(\log\log n)$~\cite{CC2009} for the
cardinality case.
The best known lower bounds on the integrality gap are $3/2$~\cite{CC2009} and 2~\cite{MWISR-LP-gap-2-Soto},
both already for the cardinality case. There is a strong connection
between the integrality gap of the LP and the maximum ratio between the coloring- and the clique-number
of a set of rectangles, see~\cite{chalermsook2011coloring}
and references therein. 

Interestingly, for the special case when all given rectangles are
squares of arbitrary sizes, the problem is much better understood.
There is a polynomial time $(1+\eps)$-approximation algorithm by Erlebach, Jansen and Seidel~\cite{EJS2001},
which works even for the more general case of arbitrary fat objects.
For the unweighted squares, and also for the more general setting
of unweighted pseudo-disks, even a simple local search algorithm gives
a PTAS~\cite{ChanHarPeled2009}. 

Although the complexity is well-understood in the setting of squares,
for rectangles it is still widely open.
In particular, the techniques of the above approximation schemes for
squares do not carry over to rectangles. The PTAS from~\cite{EJS2001}
requires that every horizontal or vertical line intersects
only a bounded number of objects of the optimal solution that are relatively large in at least
one dimension. For rectangles, this number can be up to $\Theta(n)$ which
is too much. For local search, one can easily construct examples showing
that for any size of the local search neighborhood (which gives quasi-polynomial
running time) the optimum is missed by an arbitrarily large (superconstant) factor.


For arbitrary shapes in the plane (which can be modeled as a set of
line segments) Agarwal and Mustafa~\cite{agarwal2006independent}
give an algorithm which finds an independent set of size $\sqrt{OPT/\log(2n/OPT)}$
which yields a worst case approximation factor of $n^{1/2+o(1)}$.
This was improved by Fox and Pach to $n^{\eps}$ for any $\eps>0$~\cite{FoxPach2011}.
Note that already for lines with at most one bend (i.e., lines forming an ``L'') the 
natural LP-relaxation suffers from an integrality gap of $\Omega(n)$.

To the best of our knowledge, no inapproximability result is known
for MWISR (and not even for arbitrary shapes in the 2D-plane). In
particular, an important open problem is to construct a polynomial time constant factor
approximation algorithm for MWISR.


\subsection{Our Contribution and Techniques}

We present the first $(1+\eps)$-approximation algorithm for the Maximum
Weight Independent Set of Rectangles problem with a quasi-polynomial
running time of $2^{\mathrm{poly}(\log n/\eps)}$. In contrast, the
best known polynomial time approximation algorithms achieve approximation
ratios of $O(\log n/\log\log n)$ for the weighted case~\cite{ChanHarPeled2009},
and $O(\log\log n)$ for the cardinality case~\cite{CC2009}. We
are not aware of any previous algorithms for the problem with quasi-polynomial
running time which would give better bounds than the above mentioned
polynomial time algorithms. Our quasi-PTAS rules out the possibility
that the problem is $\mathsf{APX}$-hard, assuming that $\mathsf{NP}\nsubseteq\mathsf{DTIME}(2^{\polylog(n)})$,
and thus it suggests that it should be possible to obtain significantly
better polynomial time approximation algorithms for the problem. In
addition, we present a PTAS for the case that each rectangle is $\delta$-large
in at least one dimension, i.e., if at least one of its edges has
length at least $\delta N$ for some constant $\delta>0$, assuming
that in the input only integer coordinates within $\{0,...,N\}$ occur.

Key to our results is a new geometric dynamic program \emph{GEO-DP}
whose DP-table has one entry for each axis-parallel polygon $P$ with
at most $k$ edges, where $k$ is a fixed parameter. Such a cell corresponds
to a subproblem where the input consists only of the input rectangles contained
in~$P$. The algorithm solves each such subproblem by trying every
possible subdivision of $P$ into at most $k$ polygons with again
at most $k$ edges each, and selects the partition with maximum weight
according to the DP-cells of all subproblems.

For analyzing this algorithm, we show that there is a recursive sequence
of partitions such that the rectangles of $OPT$ intersected within
that sequence have a total weight of at most $\eps\cdot OPT$. For
our QPTAS, we first provide a method to tile the plane into polygons
such that each rectangle of the optimal solution is intersected only
$O(1)$ times. Using a new stretching method for the input area, we
can guarantee that each face of our partition either contains only
rectangles of relatively small total weight, or contains at most one
rectangle. With a planar separator theorem from~\cite{Arora1998}
we can find a cut in the partition such that the intersected rectangles
have only marginal weight and both sides of the cut contain rectangles
whose total weight is upper bounded by $\frac{2}{3}OPT$. When using
these cuts in every iteration, the recursion terminates after $O(\log n / \eps)$
levels and we show that by setting $k:=(\frac{\log n}{\eps})^{O(1)}$
we obtain an approximation ratio of $1+\eps$ in quasi-polynomial
time.

We demonstrate the potential of our new algorithm by proving that
it yields a \emph{polynomial }time $(1+\eps)$-approximation algorithm
for the special case when each rectangle is large in at least one
dimension, as defined above. For this result, we employ a finer partition
of the plane which ensures that in the initial partition only large
rectangles with small total weight are intersected. Even more, each
face of the subdivision is either a path or a cycle of a small width
(strictly smaller than the longer edge of each large rectangle). Using
this, we show that GEO-DP solves each resulting subproblem within
an accuracy of $1+\eps$ by using only subpolygons with at most
a constant number of edges. Therefore, we prove that GEO-DP parametrized
by $k:=(1/\eps)(1/\delta)^{O(1)}$ gives a polynomial time $(1+\eps)$-approximation
for $\delta$-large rectangles (for any constant $\delta > 0$). In fact, this yields
a PTAS for the case that the lengths of the longer edges of the rectangles
differ by at most a constant factor (when the parameter $k$ is chosen appropriately). We would like to point out that
the initial partition might be a useful ingredient for constructing
a PTAS for the general problem since it sparsely describes the topology
of the large rectangles while losing only an $\eps$-fraction
of their total weight.

We can well imagine that our algorithmic approach finds applications
for solving the \textsc{Independent Set} problem for more general
geometric shapes. Given the large gaps in terms of approximation and
hardness results for the \textsc{Independent Set} problem in such
settings
we hope that our new techniques will help to bridge these gaps. 
Finally, we would like to note that our DP might well yield
a constant factor approximation or even a PTAS for MWISR when parametrized
by a sufficiently large parameter $k$ independent of $n$, e.g., $k=(1/\eps)^{O(1)}$.
We leave this as an open question.


\subsection{Problem Definition}
We are given a set of $n$ axis-parallel rectangles $\R=\{R_{1},...,R_{n}\}$
in the $2$-dimensional plane. Each rectangle $R_{i}$ is specified by
two opposite corners $(x_{i}^{(1)},y_{i}^{(1)}) \in\mathbb{N}^{2}$ and 
$(x_{i}^{(2)},y_{i}^{(2)})\in\mathbb{N}^{2}$, with $x_{i}^{(1)} < x_{i}^{(2)}$ and 
$y_{i}^{(1)} < y_{i}^{(2)}$, and a weight $w(R_{i})\in\mathbb{R}^{+}$. 
We define the area of a rectangle as the open set $R_{i}:=\{(x,y)|x_{i}^{(1)}<x<x_{i}^{(2)}\wedge y_{i}^{(1)}<y<y_{i}^{(2)}\}$.
The goal is to select a subset of rectangles $\R'\subseteq\R$ such that for any two
rectangles $R,R'\in\R'$ we have $R \cap R' = \emptyset$. Our objective is to maximize the total weight 
of the selected rectangles
$w(\R'):=\sum_{R\in\R'}w(R)$. 
For each rectangle $R_{i}$
we define its \emph{width }$g_{i}$ by $g_{i}:=x_{i}^{(2)}-x_{i}^{(1)}$ and its \emph{height }$h_{i}$ by $h_{i}=y_{i}^{(2)}-y_{i}^{(1)}$. 

By losing at most a (multiplicative) factor of $(1-\eps)^{-1}$ in the objective, we assume
that $1\le w(R)\le n/\eps$ for each rectangle $R\in\R$. First, we
scale the weights of all rectangles such that $\max_{R\in\R}w(R)=n/\eps$.
Since then $OPT\ge n/\eps$, all rectangles $R'$ with
$w(R')<1$ can contribute a total weight of at most $n\cdot1=\eps\cdot\frac{n}{\eps}\le\eps\cdot OPT$.
We remove them from the instance which reduces the weight of the optimal
solution by at most $\eps\cdot OPT$.


\section{\label{sec:Algorithm}The Algorithm GEO-DP}

Our results are achieved by using a new geometric dynamic programming
algorithm which we call GEO-DP and which we define in this section.
The algorithm is parametrized by a value $k\in\mathbb{N}$ which affects both
the running time and the achieved approximation ratio. In brief,
the algorithm has a DP-cell for each axis-parallel polygon $P$ with at most
$k$ edges, which represents the subproblem
consisting of all rectangles contained in $P$. When computing a
(near-optimal) solution for this subproblem, GEO-DP tries all possibilities
to subdivide $P$ into at most $k$ polygons with at most $k$
edges each and recurses.

Without loss of generality we assume that $x_{i}^{(1)},y_{i}^{(1)},x_{i}^{(2)},y_{i}^{(2)}\in\{0,...,2n-1\}$
for each $R_{i}\in\R$. If this is not the case, we transform the
instance in polynomial time into a combinatorially equivalent instance
with the latter property.

Fix a parameter $k\in\mathbb{N}$. Let $\P$ denote the set of all
polygons within the $[0,2n-1]\times[0,2n-1]$ input square whose corners have only integer coordinates and which have at most
$k$ axis-parallel edges each. Whenever we speak of a polygon, we allow
it to have holes and we do not require it to be simple. In particular, by the edges of a polygon
with holes we mean both the outer edges and the edges bounding the
holes. We introduce a DP-cell for each polygon $P\in\P$, where a cell corresponding to $P$ stores a near-optimal
solution $sol(P)\subseteq\R_P$ where $\R_P$ denotes the set of all rectangles from 
$\R$ which are contained in $P$. 

\begin{prop}\label{prop:DP-cells}
The number of DP-cells is at most $n^{O(k)}$. 
\end{prop}

\begin{figure}
\begin{centering}
\includegraphics[scale=0.75]{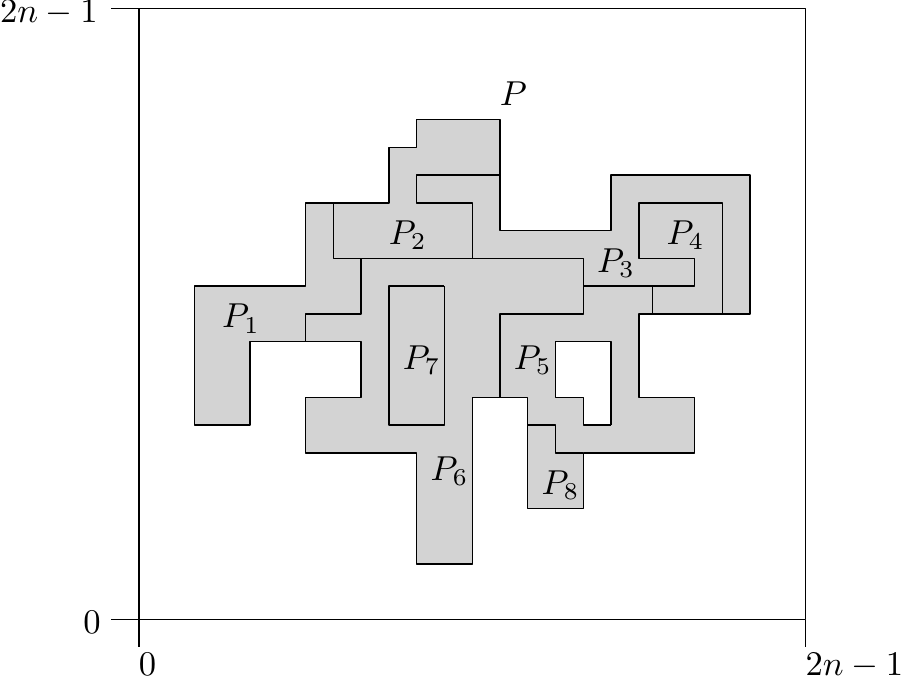}
\par\end{centering}
\caption{\label{fig:partition-P}The partition of a polygon $P$ (gray area) into at most $k$ smaller polygons, each with at most $k$ edges.}
\end{figure}

To compute the solution $sol(P)$ for some polygon $P \in \P$ we use the following
procedure. If $\R_P = \emptyset$ then we set $sol(P):=\R_P$
and terminate. Otherwise, we enumerate all possibilities to partition
$P$ into $k'$ polygons $P_{1},...,P_{k'}\in\P$ such
that $k'\le k$. See Figure~\ref{fig:partition-P} for a sketch.
Since by Proposition \ref{prop:DP-cells} we have
$|\P| \le n^{O(k)}$, the number of potential partitions we need to consider is upper bounded
by ${n^{O(k)} \choose k}=n^{O(k^{2})}$. Let $P_{1},...,P_{k'}$, where $k'\le k$, be a feasible partition (for
any enumerated set $\{P_{1},...,P_{k'}\}\subseteq\P$ this can be
verified efficiently since all polygons have axis-parallel edges with
integer coordinates in $\{0,...,2n-1\}$). For each polygon $P_{i}\in\{P_{1},...,P_{k'}\}$
we look up the DP-table value $sol(P_{i})$ and compute $\sum_{i=1}^{k'}w(sol(P_{i}))$.
We set $sol'(P):=\cup_{i=1}^{k'}sol(P_{i})$
for the partition $\{P_{1},...,P_{k'}\}$ which yields the maximum profit.
Now we define $sol(P) := sol'(P)$ if $w(sol'(P))> \max_{R\in\R_P}w(R)$,
and otherwise $sol(P):=\{R_{\max}\}$ for a rectangle $R_{\max} \in \R_P$
with maximum profit.
At the end, the algorithm outputs the value in the DP-cell which corresponds
to the polygon containing the entire input region $[0,2n-1]\times[0,2n-1]$.

Since $|\P| \le n^{O(k)}$ we get the following upper bound on the
running time of GEO-DP.

\begin{prop}\label{prop:DP-time}
\label{prop:GEO-DP-running-time}When parametrized by~$k$ the running time of GEO-DP is upper bounded by $n^{O(k^{2})}$. 
\end{prop}

For bounding the approximation ratio of GEO-DP for any parameter $k$, it is sufficient to consider only the special case that the input set $\R$ is already a feasible (optimal) solution. Therefore, we will assume this from now on.


\section{\label{sec:QPTAS}Quasi-Polynomial Time Approximation Scheme}

In this section we prove that GEO-DP achieves an approximation ratio
of $1+\eps$ when parametrized by $k=(\frac{\log n}{\eps})^{O(1)}$ and is thus a QPTAS for the MWISR problem (using Proposition~\ref{prop:DP-time}).

Key ingredient for our analysis is to show that for any set of feasible rectangles there is 
a \emph{balanced cheap cut}, i.e., a polygon which consists of only few edges, which intersects rectangles from $\R$ of marginal total weight, and which separates the rectangles from $\R$ into two parts of similar size. 
By applying such cuts recursively for $O(\log n/\eps)$ levels, we eventually obtain trivial subproblems. For proving that such good cuts always exist, we partition the plane into polygons in such a way that each rectangle is intersected only $O(1)$ times and each face of the partition consists either of exactly one rectangle or intersects rectangles of only small total weight. To ensure the latter, we apply a stretching procedure to the input before actually defining the partition. On the constructed partition we apply the weighted planar graph separator theorem from \cite{Arora1998} to obtain the cut.

\subsection{Balanced Cheap Cuts}

We introduce \emph{balanced $\alpha$-cheap $\ell$-cuts}, where $\alpha$ is a small positive value.  Intuitively, given
any set of non-overlapping rectangles $\bar{\R}$, such a cut is given
by a polygon $P$ with at most $\ell$ axis-parallel edges whose boundary
intersects rectangles with weight at most $\alpha\cdot w(\bar{\R})$
such that the interior and the exterior of $P$ each contain only
rectangles whose weight is at most ${2}/{3}\cdot w(\bar{\R})$.
 
\begin{defn}
Let $\ell\in\mathbb{N}$ and $\alpha\in\mathbb{R}$ with $0<\alpha<1$.
Let $\bar{\R}$ be a set of pairwise non-overlapping rectangles. A
polygon $P$ with axis-parallel edges is a \emph{balanced $\alpha$-cheap
$\ell$-cut} if: 
\begin{itemize}
\item $P$ has at most $\ell$ edges, 
\item for the set of all rectangles $\R'\subseteq\bar{\R}$ intersecting
the boundary of $P$ we have $w(\R')\le\alpha\cdot w(\bar{\R})$, 
\item for the set of all rectangles $\R_{\mathrm{in}}\subseteq\bar{\R}$
contained in $P$ it holds that $w(\R_{\mathrm{in}})\le2/3\cdot w(\bar{\R})$,
and 
\item for the set of all rectangles $\R_{\mathrm{out}}\subseteq\bar{\R}$
contained in the complement of $P$, i.e., in $\mathbb{R}^2 \setminus P$, it holds that $w(\R_{\mathrm{out}})\le2/3\cdot w(\bar{\R})$. 
\end{itemize}
\end{defn}

As we will show in the next lemma, GEO-DP performs well if for any
set of rectangles there exists a good cut.

\begin{lem}\label{lem:good-cut-suffices}
Let $\eps>0$. Let $\alpha>0$ with $\alpha<1/2$ and
$\ell\ge4$  be values such that for any set $\bar{\R}$ of pairwise
non-overlapping rectangles there exists a balanced\emph{ }$\alpha$-cheap
$\ell$-cut, or there is a rectangle $R\in\bar{\R}$ such that $w(R)\ge\frac{1}{3}\cdot w(\bar{\R})$.
Then GEO-DP has approximation ratio $(1+\alpha)^{O(\log(n/\eps))}$
when parametrized by $k=\ell^{2}\cdot O(\log^2(n/\eps))$. 
\end{lem}

\begin{proof}[Proof (sketch). ]
Starting with the $[0,2n-1]\times[0,2n-1]$ input square, we either find
a rectangle $R$ with $w(R)\ge\frac{1}{3}\cdot w(\bar{\R})$, or a
balanced $\alpha$-cheap $\ell$-cut. 
In either case we get a decomposition of the problem into subproblems, where each subproblem is defined by a polygon 
which contains rectangles whose total weight is at most
a $2/3$-fraction of $w(\R)$. Continuing for 
$O(\log n/\eps)$ recursion levels, we obtain subproblems consisting
of at most one rectangle each (as $1\le w(R)\le n/\eps$
for each $R\in\R$). 
Each appearing subproblem can be expressed
as the intersection of $O(\log n/\eps)$ polygons with at
most $\ell+4$ edges each (we have to add $4$ since the boundary of the input square can become the outer boundary of a polygon). Thus, the boundary of each considered subproblem consists of $O(\ell^2 \cdot \log^2(n/\eps))$ edges. Additionally, we can show that each subproblem gives rise to $O(\ell^2 \cdot \log^2(n/\eps))$ subproblems in the next recursion level.
As GEO-DP tries all partitions of a
polygon into at most $k$ polygons, each with at most $k$ edges, at each recursion level it will consider the partition corresponding to the cut. 
Since at each level we lose rectangles whose weight is at most an $\alpha$-fraction of the total weight of the rectangles from the current recursion level, we obtain the claimed approximation ratio.
\end{proof}

In the remainder of this section we prove that for any set $\bar{\R}$ and any $\delta >0$
there is a balanced $O(\delta)$-cheap $\left(\frac{1}{\delta}\right)^{O(1)}$-cut
or there is a rectangle $R\in\bar{\R}$ with $w(R)\ge\frac{1}{3}\cdot w(\bar{\R})$. This implies our main result
when choosing $\delta:=\Theta(\eps/ \log(n/\eps))$.
Note that from now on our reasoning does not need to be (algorithmically) constructive.

\subsection{Stretching the Rectangles}

For the purpose of finding a good cut, we are free to stretch or squeeze
the rectangles of $\R$. 
We do this as a preprocessing step in order
to make them well-distributed.

\begin{defn}
A set of rectangles $\Rb$ with integer coordinates in $\{0,...,N\}$
is \emph{well-distributed }if for any $\gamma>0$ and for any $t\in\{0,...,N\}$
we have that all rectangles contained in the area $[0,N]\times[t,t+\gamma\cdot N]$
have a total weight of at most $2\gamma\cdot w(\R)$. We require
the same for the rectangles contained in the area $[t,t+\gamma\cdot N]\times[0,N]$. 
\end{defn}

We say that two sets of rectangles $\R, \Rb$ are \emph{combinatorially
equivalent} if we can obtain one from the other by stretching or squeezing the input area, possibly non-uniformly (see Appendix~\ref{apx:QPTAS} for a formal definition).

\begin{lem}\label{lem:weighted-stretching}
Let $\R$ be a set of rectangles with arbitrary weights.
There is a combinatorially equivalent set $\bar{\R}$ using only integer coordinates
in $\{0,...,4 \cdot |\R| \}$ which is well-distributed.
\end{lem}

\begin{proof}[Proof (sketch). ]
Without loss of generality we assume that $x_{i}^{(1)},y_{i}^{(1)},x_{i}^{(2)},y_{i}^{(2)}\in\{0,...,2n-1\}$ for each $R_{i}\in\R$, where $n=|\R|$. We stretch the original input square in such a way that the lengths of its sides double (i.e., increase by another $2|\R|$), and the distance between two original consecutive $x$-coordinates ($y$-coordinates) $x_i$ and $x_i+1$ increases proportionally to the weight of all rectangles "starting" at the $x$-coordinate ($y$-coordinate) $x_i$. We then need to introduce some rounding, as we want the new coordinates of all the rectangles to be integral. The new set of rectangles is clearly equivalent to the original one. 

Consider any vertical stripe $S$ in the modified input square with left- and rightmost $x$-coordinates $x'_a$ and $x'_b$, respectively. 
There is a set of $x$-coordinates $x_1,x_2,...,x_m$ from the original instance which were mapped to values $x'_1,x'_2,...,x'_m \in [x'_a,x'_b]$. All rectangles contained in $S$ must have their respective leftmost $x$-coordinates in $\{x'_1,x'_2,...,x'_{m-1}\}$. The total weight of rectangles with leftmost coordinate $x'_i$ is proportional to $x'_{i+1}-x'_i$, for each $i$. Hence, the total weight of rectangles contained in $S$ is proportional to $x'_m-x'_1 \le x'_b-x'_a$. 
The same is true for horizontal stripes, and so the modified input instance is well-distributed.
\end{proof}

Observe that there is a balanced $\alpha$-cheap $\ell$-cut for any values $\alpha$ and $\ell$ in the stretched instance if and only if there is such a cut in the original instance.
Thus, suppose from now on that we have a well-distributed set of pairwise non-intersecting rectangles $\R$, using integer coordinates
in $\{0,...,N\}$ for some integer $N$, and a value $\delta>0$ such that $1/\delta \in \mathbb{N}$. As we do not require any special bound on the value of $N$, we can scale up all coordinates of the rectangles by a factor of $(1/ \delta)^2$, and therefore we can assume that $\delta^2 N$ is an integer. 


\subsection{Partitioning the Plane}

We define a procedure to partition the input square $I:= [0,N] \times [0,N]$. This partition is defined by only $(1/\delta)^{O(1)}$ lines, and it has the properties that each rectangle in $\R$ is intersected only $O(1)$ times and each face either surrounds exactly one rectangle or it has non-empty intersection with rectangles with small total weight of at most $O(\delta^2 w(\R))$.

We call a rectangle $R_{i}\in\R$ \emph{large} if $h_{i}>\delta^{2} N$ or $g_{i}>\delta^{2} N$, and \emph{small} if $h_{i} \le \delta^{2} N$ and $g_{i} \le \delta^{2} N$. We denote the subsets of $\R$ consisting of large and small rectangles by $\R_L$ and $\R_S$, respectively. 
We call a rectangle $\R_i \in \R$ \emph{vertical} if $h_i > g_i$, and \emph{horizontal} if $h_i \le g_i$. We say that a line $L$ \emph{cuts} a rectangle $R \in \R$, if $R \setminus L$ has two connected components.

We now present the construction of the partition. It will consist of a set of straight axis-parallel lines $\L$  
contained in the input square $I$, and containing the boundary of $I$. A connected component of $I \setminus \L$ is called a \emph{face}, and the set of faces is denoted by $\F(\L)$. Note that the faces are open polygons and for any $F \in \F(\L)$ and $L \in \L$ we have $F \cap L = \emptyset$.

\paragraph{Grid.}
We subdivide the input square $I$ into $\frac{1}{\delta^{2}}\times\frac{1}{\delta^{2}}$
\emph{grid cells}, where each grid cell is a square of size $\delta^{2} N \times \delta^{2} N$. 
Formally, for each $i,j\in \{0,...,1/\delta^2-1\}$ we have a grid cell $[i\cdot \delta^{2} N,(i+1)\cdot \delta^{2} N] \times    [j\cdot \delta^{2} N,(j+1)\cdot \delta^{2} N]$.
As $\delta^2 N$ is an integer, the corners of the grid cells have integer coordinates. The lines subdividing the input square into the grid cells are called \emph{grid lines}.

We say that a rectangle $R \in \R$ \emph{intersects} a grid cell $Q$, if $R \cap Q \neq \emptyset$ (recall that rectangles have been defined as open sets). Each rectangle $R \in \R_L$ intersects at least two grid cells, and each rectangle $R \in \R_S$ intersects at most four grid cells. We say that a rectangle $R \in \R$ \emph{crosses} a grid cell $Q$, if $R$ intersects $Q$ and $R$ has non-empty intersection with two opposite edges of $Q$. Notice that small rectangles do not cross any grid cells.

\paragraph{Rectangle faces.}
For each large vertical rectangle which is cut by a vertical grid line, and for each large horizontal rectangle which is cut by a horizontal grid line, we add the edges of the rectangle to the set of lines $\L$ (see Figure \ref{fig:maze}a). The added edges are called \emph{rectangle edges}, and the faces corresponding to such rectangles are called \emph{rectangle faces}.

\begin{figure}
\begin{centering}

\centerline{
\includegraphics[width=0.3\textwidth]{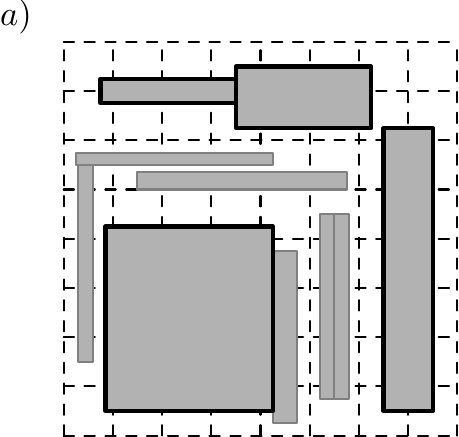}
\hskip0.1\textwidth
\includegraphics[width=0.3\textwidth]{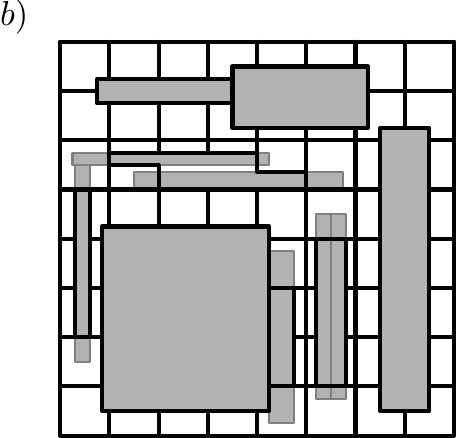}
}
\par\end{centering}
\caption{\label{fig:maze}Creating the partition of the input square. The large rectangles $\R_L$ are depicted in gray. The bold edges represent the lines from $\L$.}
\end{figure}

\begin{lem}\label{lem:rectangle-faces}
The number of rectangle faces is at most $2(1/\delta)^4$.
\end{lem}  

Notice that if a large vertical rectangle is not contained in a rectangle face then it is contained in a single column of grid cells. 
Similarly, if a large horizontal rectangle is not contained in a rectangle face then it is contained in a single row of grid cells.

\paragraph{Lines within the grid cells.}
We now consider each grid cell $Q$ separately, and proceed as follows. Let $\R_Q$ denote the set of rectangles from $\R$ which are not contained in the rectangle faces, and which cross $Q$. Notice that, as the rectangles from $\R$ are pairwise non-overlapping, $\R_Q$ cannot contain both vertical and horizontal rectangles.
\begin{itemize}
\item If $\R_Q = \emptyset$, we add to $\L$ the whole boundary of $Q$, with the exception of the fragments which are in the interior of the rectangle faces  (see Figure \ref{fig:maze-in-cell}a). 

Notice that the boundaries of the rectangle faces are in $\L$, so in this case $\L \cap Q$ is connected.

\item If $\R_Q$ consists of vertical rectangles, let $L_{\ell}$ and $L_r$ be the leftmost and the rightmost vertical edge of a rectangle from $\R_Q$. We add to $\L$ the lines $L_{\ell} \cap Q$ and $L_r \cap Q$, and the boundary of $Q$ with the exception of the fragments which are between $L_{\ell}$ and $L_r$, or in the interior of the rectangle faces (see Figure \ref{fig:maze-in-cell}b).

Notice that the lines added to $\L$ while considering the grid cell $Q$ do not intersect any rectangles from $\R_Q$.

\item If $\R_Q$ consists of horizontal rectangles, we proceed as before, considering horizontal lines instead of vertical (see Figure \ref{fig:maze-in-cell}c).
\end{itemize}

\begin{figure}
\begin{centering}
\includegraphics[width=0.7\textwidth]{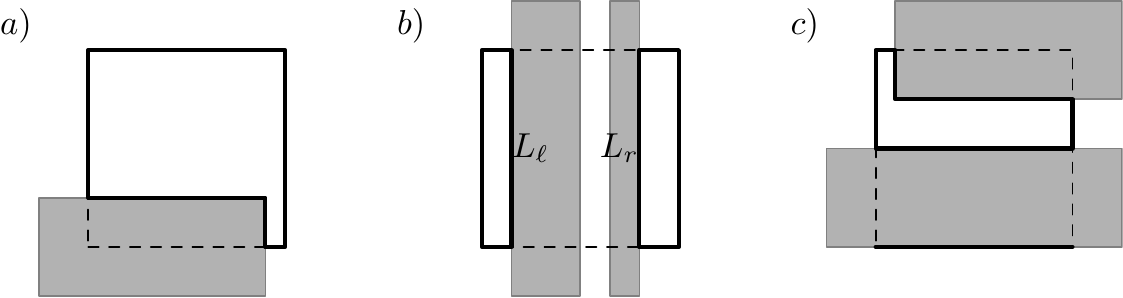}
\par\end{centering}
\caption{\label{fig:maze-in-cell}Constructing the lines within a single grid cell $Q$. The depicted gray rectangles are either in $\R_Q$ or they are rectangles from $\R_L$ which have their own respective rectangle face. The bold lines correspond to the lines of $\L$ within the grid cell.
}
\end{figure}

Notice that the lines from $\L$ can \emph{overlap} (i.e., we can have $L_1, L_2 \in \L$ s.t. $L_1 \cap L_2$ is an interval), but they do not \emph{intersect properly} (i.e., if for $L_1, L_2 \in \L$ we have $L_1 \cap L_2 = \{p\}$, then $p$ is an endpoint of at least one of the lines $L_1, L_2$). The lines from $\L$ cover the boundary of the input square $I$. An example of the partition can be seen in Figure \ref{fig:maze}b.

\paragraph{Graph $G(\L)$.} We now construct a graph $G(\L) = (V,E)$ embedded in the input square $I$, representing the partition $\L$. Any point $p \in I$ becomes a vertex of $G(\L)$ if and only if there is at least one line $L \in \L$ with an endpoint in $p$. For any pair of vertices $v,w \in V$ for which there is a line $L \in \L$ such that $\{v,w\} \in L$, and for which no vertex lies on the straight line strictly between $v$ and $w$,
we add an edge $vw$ to $G(\L)$ (i.e., edges of $G(\L)$ represent subdivisions of lines in $\L$). As $\bigcup_{L \in \L} L = \bigcup_{e \in E} e$, the faces of $G(\L)$ are exactly $\F(\L)$.
The claim of the next lemma is directly implied by the construction.

\begin{lem}\label{lem:G-some-properites}
The graph $G(\L)$ is planar and has $O((1/\delta)^4)$ vertices and $O((1/\delta)^4)$ edges.
\end{lem}

The following lemmas will be needed to show the existence of a balanced cut in $G(\L)$.

\begin{lem}\label{lem:rec-intersect-constant}
Each rectangle from $\R$ can be intersected by at most four edges of the graph $G(\L)$.
\end{lem}
\begin{proof}[Proof (sketch). ] 
The only edges of $G(\L)$ intersecting rectangles from $\R$ lie on grid cell boundaries. We can show that any rectangle $R \in \R$ can be intersected by at most one edge of $G(\L)$ at each grid cell boundary. 

As any $R \in \R_S$ intersects at most four grid cell boundaries, the lemma holds for small rectangles. The lemma clearly holds for any rectangle $R \in \R_L$ contained in a rectangle face. The remaining case are large rectangles contained in a single row or column of grid cells. From the construction of $\L$ for single grid cells we can show that such a rectangle $R \in \R_L$ can be intersected only at the two extremal grid cell boundaries within $R$. 
\end{proof}

\begin{lem}\label{lem:rec-intersect-faces}
Let $F \in \F(\L)$. The boundary of $F$ intersects rectangles from $\R$ of total weight at most $8 \delta^2 w(\R)$.
If $F$ is not a rectangle face, then $F$ has non-empty intersection with rectangles from $\R$ of total weight at most  $ 8\delta^2 w(\R)$.
\end{lem}
\begin{proof}[Proof (sketch). ]
The lemma clearly holds for rectangle faces. Let $F \in \F(\L)$ be face which is not a rectangle face. If $F$ is contained in one grid cell $Q$
then one can show that all rectangles intersecting $F$ must be contained in the area defined by grid column and grid row containing $Q$ together with the two adjacent grid rows (see Figure~\ref{fig:stripes}a). On the other hand, if $F$ spans several grid cells, one can show that all rectangles intersecting it must be contained in a single grid row or column  (see Figure~\ref{fig:stripes}b). In both cases, the claim follows since $\R$ is well-distributed.
\end{proof}


\subsection{Defining the Cut}

For obtaining our desired cut, we apply the following theorem from
\cite{Arora1998} for the graph $G(\L)$. A \emph{V-cycle} $C$ is a Jordan curve in the embedding of a given
planar graph $G$ which might go along the edges of $G$ and also might
cross faces of $G$. The parts of $C$ crossing an entire face of
$G$ are called \emph{face edges.}

\begin{thm}[\cite{Arora1998}]\label{thm:cycle-separator}
\label{thm:V-cycle}Let $G$ denote a planar, embedded graph with weights on the vertices and faces and
with costs on the edges. Let $W$ denote the total weight, and
$M$ the total cost of the graph. Then, for any parameter $\bar{k}$, we can find in polynomial
time a separating V-cycle $C$ such that
\begin{itemize}
\item the interior and exterior of $C$ each has weight at most $2W/3$,
\item $C$ uses at most $\bar{k}$ face edges, and
\item $C$ uses ordinary edges of total cost $O(M/\bar{k})$.
\end{itemize}
\end{thm}

First, we need to assign costs to the edges of $G(\L)$ and weights to the vertices and faces of $G(\L)$.
For each edge $e\in E$ we define its cost $c_{e}$ to be the total weight of rectangles intersecting $e$. 
The weights of all vertices are zero. For each face $F$ we define its weight
$w_{F}$ to be the total weight of all rectangles contained in $F$, plus a fraction of the weight of the rectangles which intersect the boundary of $F$. If a rectangle $R \in \R$ has non-empty intersection with $m$ faces, each of these faces obtains a $1/m$-fraction of the weight of $R$. From Lemmas \ref{lem:rec-intersect-constant} and \ref{lem:rec-intersect-faces} we obtain the following bounds.
 
\begin{lem}\label{lem:G-cost-weight}
The total cost of edges in $G(\L)$ is at most $4 w(\R)$. The weight of each non-rectangle 
face $F$ is at most $8 \delta^{2}\cdot w(\R)$. The total weight of the faces equals $w(\R)$. 
\end{lem}

For constructing the cut we apply Theorem \ref{thm:cycle-separator} with parameter $\bar{k}=1/\delta$ to the graph $G(\L)$. The obtained  V-cycle $C$ yields a cut in the plane. We replace each face edge crossing some face $F \in \F(\L)$ by the edges going along the boundary of~$F$. If $w(R)\le  w(\R)/3$ for each rectangle $R \in \R$ and if $\delta < 1/5$, then, from Lemma \ref{lem:rec-intersect-faces}, each face has weight at most $w(\R)/3$. 
We can then ensure that each side of the modified cut contains rectangles of total weight at most $2 w(\R)/3$. Using the upper bound on the number of edges of $G(\L)$ from Lemma \ref{lem:G-some-properites}, and upper bounding the total weight of rectangles intersected by the modified cut (using Lemmas \ref{lem:rec-intersect-faces} and \ref{lem:G-cost-weight}) implies the following result.

\begin{lem}\label{lem:balanced-cut-parametrized}
Assume that $1/5 > \delta>0$. For any set $\R$ of pairwise non-overlapping rectangles not containing a rectangle of weight at least $w(\R)/3$
there exists a balanced $O(\delta)$-cheap $O((1 / \delta)^4)$-cut.
\end{lem}

When choosing $\delta:=\Theta(\eps/ \log(n/\eps))$, from Lemma \ref{lem:good-cut-suffices} and Lemma \ref{lem:balanced-cut-parametrized}  we obtain that GEO-DP is a QPTAS.

\begin{thm}\label{thm:qptas}
The algorithm GEO-DP parametrized by $k=(\frac{\log n}{\eps})^{O(1)}$ yields a quasi-polynomial time approximation scheme for the maximum weight independent set of rectangles problem.
\end{thm}


\section{\label{sec:PTAS-large-rectangles}A PTAS for Large Rectangles}

In this section we show that GEO-DP yields a polynomial time approximation scheme for
input instances which contain only large rectangles, i.e., in
which every rectangle has width or height greater than a $\delta$-fraction
of the length of the edges of the input square, for some constant $\delta > 0$. As a corollary, we obtain a PTAS
for the special case of the MWISR problem when the lengths of the longer edges the rectangles differ only by a constant factor, i.e., for some constant $\delta > 0$ we have $\max\{h_i,g_i\}\le 1/\delta \cdot \max\{h_{i'},g_{i'}\}$ for all rectangles $R_i, R_{i'}$.

Let $\eps>0$ and $\delta>0$. Let $\R$ be a set of rectangles, and let $N$ be an integer such that for each rectangle $R_{i}\in\R$ we have $x_{i}^{(1)},x_{i}^{(2)},y_{i}^{(1)},y_{i}^{(2)}\in\{0,...,N\}$. We call a rectangle $R_{i}\in\R$  \emph{$\delta$-large} if $h_{i}>\delta N$ or $g_{i}>\delta N$. 
In this section we assume that the input consists of a set $\R$ of $\delta$-large rectangles for some constant $\delta >0$. Assume
w.l.o.g.~that $1/\delta\in\mathbb{N}$ and $\delta N\in\mathbb{N}$. As in the previous section, for the analysis of GEO-DP we can assume that $\R$ itself is the optimal solution, i.e., no two rectangles in $\R$ overlap.

\paragraph{Overview.}
First, we show that there is a way to partition the plane using a set of at most $\frac{1}{\eps} \cdot (\frac{1}{\delta})^{O(1)}$ lines, such that the intersected rectangles have small total weight and each face of the partition is a path or a cycle of ``width'' at most $\delta N$. Note that the latter bound is strictly smaller than the length of the longer edge of each rectangle. In a sense, this partition sparsely describes the topology of the (large) rectangles while losing only rectangles of negligible weight.
Then, we show that for each face GEO-DP can solve the resulting subproblem within a $(1+\eps)$-accuracy, without an increase in the complexity of the subproblems during the recursion.

When given an input instance, GEO-DP first preprocesses it so that all rectangles have coordinates which are integers in~$\{0,...,2n-1\}$. Note, however, that this routine might cause that some rectangles are not $\delta$-large anymore. Therefore, in the analysis in this section, we show that a good recursive subdivision of the input square exists for the \emph{original} input with coordinates in $\{0,...,N\}$ for some integer $N$, and where all the rectangles are $\delta$-large. As the preprocessing consists essentially of stretching and squeezing of the input area, there is a corresponding recursive subdivision of the preprocessed input instance, whose polygons have the same complexity, and which will be considered by GEO-DP. W.l.o.g.~we assume that $\delta N$ is an integer.


\subsection{\label{sub:Constructing-the-Maze-for_Large}Constructing the Partition for Large Rectangles}

We define a set of lines $\L$ forming
a partition in the input square $[0,N]\times[0,N]$. 
The lines in $\L$ will have the properties that 
\begin{itemize}
\item $|\L|\le\frac{1}{\eps} \cdot (\frac{1}{\delta})^{O(1)}$, 
\item the rectangles intersected by a line in $\L$ have a total weight
of at most $\eps \cdot w(\R)$, and
\item each face in the partition obtained by $\L$ which contains rectangles from $\R$ is either a path or a cycle with ``width'' at most $\delta N$.
\end{itemize}
Without saying explicitly, from now on each considered line is either
horizontal or vertical and its endpoints have integral coordinates.

\paragraph{Grid and blocks.}
We construct a grid consisting of $\frac{1}{\delta}\times\frac{1}{\delta}$
grid cells in the input square $[0,N]\times[0,N]$,
i.e., for each $i,j\in\{0,...,1/\delta-1\}$ there is a grid
cell with coordinates $[i\cdot\delta N,(i+1)\cdot\delta N]\times[j\cdot\delta N,(j+1)\cdot\delta N]$. 

We slice all rectangles parallel to their longer edge into \emph{blocks}, i.e., rectangles of unit width or height. Formally, we cut each rectangle $R_{i}\in\R$ with $h_{i}>g_{i}$ into $x_{i}^{(2)}-x_{i}^{(1)}$ \emph{vertical blocks}, with the corners $(j,y_{i}^{(1)})$ and $(j+1,y_{i}^{(2)})$ for $j=\{x_{i}^{(1)},x_{i}^{(1)}+1,\ldots,x_{i}^{(2)}-1\}$. With a symmetric operation we generate \emph{horizontal blocks }for each rectangle $R_{i}\in\R$ with $h_{i}\le g_{i}$. We denote by $\B$ the set of all generated blocks and observe that also they have integer coordinates. Like the rectangles, we define the blocks as open sets.
We will first find a partition for the blocks, which essentially means that in the first version of the partition we can cut the rectangles arbitrarily parallel to their longer edges. Later we will show how to adjust the partition---by introducing some detours---so that these cuts will be eliminated.

We use the following notation.
A line $L$ \emph{touches} a rectangle $R$ if $L\cap (R \cup \partial R)\ne\emptyset$. A line $L$ \emph{intersects} a rectangle $R$ if $L\cap R\ne\emptyset$. A line $L$ \emph{hits} a rectangle $R$ if $L$ touches $R$, $L$ does not intersect $R$, but extending $L$ would result in $L$ intersecting $R$. A line $L$ \emph{cuts} a rectangle $R$ if $R \setminus L$ has two connected components.
We say that a rectangle $R$ (a block $B$, a line $L$) \emph{intersects} a grid
cell $Q$ if $R\cap\int(Q)\neq\emptyset$ ($B\cap\int(Q)\neq\emptyset$, $L\cap\int(Q)\neq\emptyset$).
Each rectangle $R\in\R$, as well as each block $B\in\B$, intersects
at least two grid cells. We say that a block $B\in\B$ \emph{ends}
in a grid cell $Q$, if $B$ intersects $Q$, and for a short edge
$e$ (i.e., an edge with unit length) of $B$ we have $e\subseteq Q$.

\paragraph*{Initial set of lines.}

We start by introducing an initial set of lines $\L_{0}$ as follows. First, we add to $\L_{0}$ four lines which form the boundary of the input square $[0,N] \times [0,N]$.

Consider a grid cell $Q$ and its bottom edge $e$. If possible, we
add to $\L_{0}$ the following maximal lines which do not intersect
any block or any line previously added to $\L_{0}$, which touch $e$,
and which are strictly longer than $\delta N$: 
\begin{itemize}
\item a vertical line with the smallest possible $x$-coordinate, 
\item a vertical line with the largest possible $x$-coordinate, 
\item a vertical line $L$ which maximizes the length of $L\cap Q$. If
there are several such lines, we add two: one with the smallest and
one with the largest $x$-coordinate. Lines maximizing $|L\cap Q|$ are called \emph{sticking-in lines} for $e$ in $Q$. The ones added to $\L_{0}$ in this step are called \emph{extremal sticking-in lines}.
\end{itemize}
\begin{figure}
\begin{centering}
\includegraphics{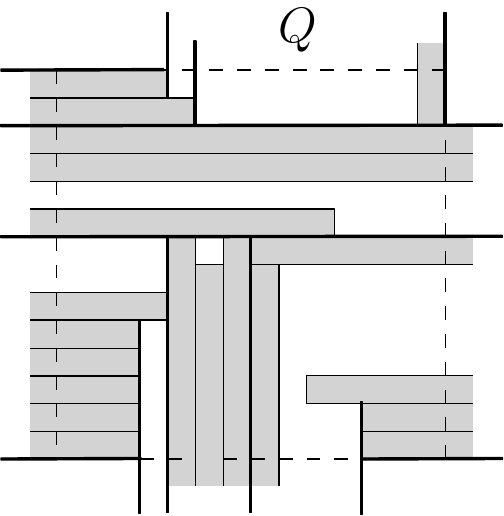} 
\par\end{centering}

\caption{\label{fig:construction-lines-L}The thick lines denote the lines in $\L_0$
added for the grid cell $Q$. The blocks of the considered instance are depicted in gray.}
\end{figure}

We do the same operation for the top, left and right edges of $Q$,
where for the left and right edges we take horizontal lines, considering
the $y$-coordinates instead of the $x$-coordinates. We do this in
a fixed order, e.g., first we add all vertical lines, and then all
horizontal lines. See Figure~\ref{fig:construction-lines-L} for an
example. We do not want $\L_{0}$ to be a multi-set and thus we add each
line at most once. Note that for any two lines $L_{1},L_{2}\in\L_{0}$
the intersection $L_{1}\cap L_{2}$ is either empty or consists of
one single point (which is the endpoint of one of the lines). All lines
in $\L_{0}$ are maximal, which means that they cannot be extended
without intersecting any perpendicular block or a perpendicular line
in~$\L_{0}$.
\begin{prop}\label{prop:l-size}
The set $\L_{0}$ consists of at most $16(\frac{1}{\delta})^{2}+4$
lines.
\end{prop}

\paragraph*{Extending lines.}

A line in $\L_{0}$ might have \emph{loose ends} which are endpoints
which are not contained in some other line in $\L_{0}$. We fix this
by adding a set of lines $\Le$. We extend each loose end $p$ of
a line in $\L_{0}$ by a path connecting $p$ either to a line in
$\L_{0}$ or to a line in the so far computed set $\Le$. Such
a path will contain $O({1}/(\eps {\delta}^{2}))$ horizontal
or vertical line segments, and will cut only rectangles of total weight $O(\eps \delta^2 \cdot w(\R))$ parallel to their shorter edges. 
We add the lines of this path to $\Le$ and continue with the next
loose end of a line in $\L_{0}$. 

The details of the construction can be found in Appendix \ref{app:construction-le}. Here we present just the idea of the construction. As all lines from $\L_{0}$ are maximal, if an endpoint of a line $L \in \L_{0}$ does not hit a line from $\L_{0}$ then it must hit a perpendicular block $B \in \B$. The first line $L_1$ on the path goes along the boundary of $B$ such that it crosses the boundary of a grid cell (as blocks are not contained in a single grid cell, their longer edges always cross the boundary of a grid cell). We extend $L_1$ so that it either touches a line from $\L_{0} \cup \Le$ (and the construction of the path is finished), or hits some perpendicular block. In the latter case we proceed with the construction of the path, considering a loose end of $L_1$ instead of $L$.
Notice that in this part of the construction the path does not intersect any rectangles parallel to their shorter edges (i.e., it does not intersect any blocks). After $O({1}/(\eps {\delta}^{2}))$ steps either the path ends by touching a line from $\L_{0} \cup \Le$, or we can make a shortcut by adding a line $L'$ at the end of the path such that $L'$ goes along a grid cell boundary, connects the path with a line from $\L_{0} \cup \Le$, and cuts rectangles of total weight $O(\eps \delta^2 \cdot w(\R))$ parallel to their shorter edges. See Figure~\ref{fig:construction-lines-Le} for an example of the construction.

\begin{figure}
\begin{centering}
\includegraphics[width=0.3\textwidth]{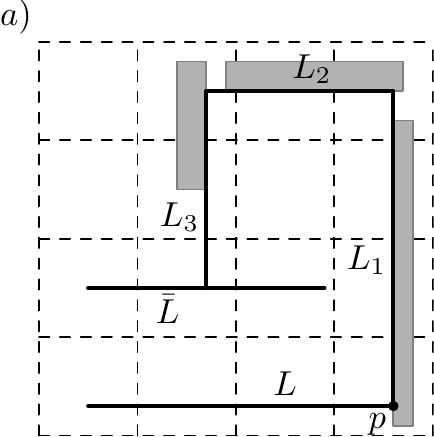}
\hskip0.1\textwidth
\includegraphics[width=0.3\textwidth]{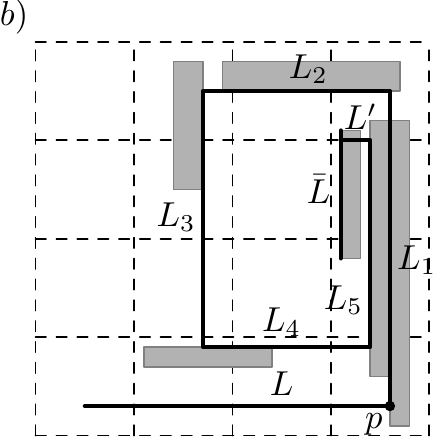} 
\par\end{centering}

\caption{\label{fig:construction-lines-Le}The construction of the lines $\Le$. A new path is constructed, starting at a loose end $p$ of a line $L \in \L_{0}$. In case $a)$ the construction of the path ends when $L_3$ hits a line $\bar{L} \in \L_{0}$. In case $b)$ we make a "cheap shortcut", by ending the path $L_1, \ldots , L_5$ with a line $L'$, which connects $L_5$ with a line $\bar{L} \in \L_{0}$. $L'$ cuts rectangles of small total weight.}
\end{figure}

We say that a set of lines $\L$ is \emph{nicely connected
}if no two lines $L,L'\in\L$ overlap (i.e., share more than one point)
or intersect properly (i.e., such that $L\cup L'\setminus\{L\cap L'\}$
has four connected components) and for any endpoint $p$ of a line
$L\in\L$ there is a line $L'\in\L$, perpendicular to $L$, such
that $L\cap L'=\{p\}$.

\begin{lem}\label{l-and-le-properties}
The set of lines $\L_{0}\cup\Le$ is nicely
connected, $|\Le|\le\frac{1}{\eps}\cdot(\frac{1}{\delta})^{O(1)}$, and the total weight
of rectangles in $\R$ cut by some line in $\L_{0}\cup\Le$ parallel to their shorter edge
is upper bounded by $\eps\cdot w(\R)$. Also, all lines
in $\Le$ cutting rectangles in $\R$ lie on some grid line and for each
line $L \in \Le$ there exists no cell $Q$ such that $L \subseteq \int(Q)$.
\end{lem}

\paragraph{Faces of the partition.}
The lines $\L_{0} \cup \Le$ subdivide the input square into a set of faces which are the connected components of $I\setminus (\L_{0} \cup \Le)$ (so in particular, the faces are open sets). Denote by $\F(\L_{0} \cup \Le)$
the set of all faces of this partition, and by $\F_+(\L_{0} \cup \Le)$ the set of all faces which contain at least one block from $\B$. As the next lemma shows, inside of each grid cell each face from the set $\F_+(\L_{0} \cup \Le)$ has a simple structure.
We say that a polygon $P$ is an \emph{L-shape} if its boundary has exactly six axis-parallel edges.

\begin{lem}\label{lem:corridor-shapes} 
Consider a face $F\in\F_+(\L_{0}\cup\Le)$ 
and let $Q$ be a grid cell with $F\cap Q\ne\emptyset$. Consider one connected component $C$
of $F\cap Q$. Then $\int(C)$ is the interior of a rectangle or the
interior of an L-shape. Also, $C\cap\partial Q$ consists of one or
two disjoint lines, where $\partial Q$ denotes the boundary of~$Q$.
\end{lem}
\begin{proof}[Proof (sketch).]
First, assume that there is an edge $e$ of $Q$ and two lines $L,L' \in \L_{0}\cup\Le$ from the boundary of $C$ such that $L \cap e \ne \emptyset \ne L' \cap e$ and the subsegment of $e$ between $L \cap e$ and $L' \cap e$ is contained in $C \cup \partial C$.
Using that the extremal sticking-in lines of $Q$ belong to $\L_0$, with some careful analysis we can show that there is an edge $e' \ne e$ of $Q$ 
and two lines $\bar{L},\bar{L'}$ connecting $L$ and $L'$ with $e'$, respectively (possibly $L=\bar{L}$ and $L'=\bar{L'}$). With this insight and the fact that the lines in $\L_{0}\cup\Le$ are nicely connected, we can show that if $e'$ is opposite of $e$ then $\int(C)$ is the interior of a rectangle, otherwise $\int(C)$ is the interior of an L-shape.

Next, we need to show that an edge $e$ and lines $L,L' \in \L_{0}\cup\Le$ with the needed properties always exist. First, we show that they exist if $C$ has a non-empty intersection with some block $B \in \B$ contained in $F$. That holds, as for some edge $e$ of $Q$ we have $B \cap e \neq \emptyset$ and then $B \cap e  \subseteq C$. As the extremal long lines crossing $e$ are in $\L_{0}$, there are lines $L,L'$ with the properties above which ``surround'' $B$ within $Q$, and they form the boundary of $C$.
Last, we show that if we have two neighboring connected components $C$ and $C'$ and the needed properties hold for one of them, then they must hold for the other one as well.
\end{proof}

Now we study the structure of the faces in $\F_+(\L_{0}\cup\Le)$ at the boundary of the grid cells. In the following lemma we show that multiple connected components of a face inside one grid cell $Q'$ cannot merge into one component in a neighboring grid cell $Q$.  

\begin{lem} \label{lem:corridor-at-boundaries}
Let $Q$ and $Q'$ be grid cells such that $Q \cap Q' = \{e\}$ for an edge $e$.
Consider a face $F\in\F_+(\L_{0}\cup\Le)$ 
such that $F\cap Q\ne\emptyset$, and let $C$ be a connected component of $F\cap Q$ such that $C\cap e\ne\emptyset$.
Then there is exactly one connected component $C'$ of $F\cap Q'$ such that $C \cap C' \neq \emptyset$.
\end{lem}
\begin{proof}[Proof (sketch).]
Clearly, at least one such component $C'$ exists since $F\cap e\ne\emptyset$. Assume for contradiction that there are two connected components $C'_1$ and $C'_2$ of $F\cap Q'$ with non-empty intersection with $C$. Then there must be a line $L \in \L_{0} \cup \Le$ intersecting $Q'$ which touches $e$ between $C'_1 \cap e$ and $C'_2 \cap e$. We can show, from the construction of the lines in $\L_{0} \cup \Le$, that $L$ is connected with a boundary of $C$ in $Q$ via a line from $\L_{0} \cup \Le$. This is a contradiction, as then either $C'_1$ or $C'_2$ does not intersect $C$.
\end{proof}

\paragraph*{Circumventing some rectangles.}As the last step of the construction of the partition, we want to ensure that if a line $L$ in
our construction intersects a rectangle $R \in \R$, then it cuts $R$
parallel to its short edge.
We achieve this as follows: whenever a line $L \in \L_{0}\cup\Le$ 
intersects a rectangle $R\in\R$ such that $R\setminus L$ has only one
connected component or $R\cap L$ is longer than $\delta N$,
then we add the four edges of $R$ as new lines and remove all parts
of lines from $\L_{0}\cup\Le$ 
which are inside $R$. For any line in $\L_{0}\cup\Le$ this operation adds at most $O(1/\delta)$ new lines.
Denote by $\L$ the resulting
final set of lines. Similarly as above, the set $\F(\L)$ denotes all faces, and the set $\F_+(\L)$ denotes all faces which contain at least
one rectangle.

Using the bounds on the number of edges in $\L_{0} \cup \Le$ from Proposition \ref{prop:l-size} and Lemma \ref{l-and-le-properties}, and the upper bound on the total weight
of rectangles in $\R$ cut by a line from $\L_{0}\cup\Le$ parallel to its 
shorter edge (Lemma \ref{l-and-le-properties}), we can show the following result.

\begin{lem}\label{lem:maze-properties}
The set of lines $\L$ has the properties that $|\L|\le\frac{1}{\eps}\cdot\left(\frac{1}{\delta}\right)^{O(1)}$ and
the total weight of intersected rectangles is upper bounded
by $\eps \cdot w(\R)$.
\end{lem}


\subsection{Solving the Subproblems for the Faces}
We transform the set of lines $\L$ into a graph $G(\L) = (V,E)$ in the same way as in Section \ref{sec:QPTAS}. From 
Lemma~\ref{lem:maze-properties} we get that $|V| \le (1/\eps) ({1}/{\delta})^{O(1)}$ and $|E| \le (1/\eps) ({1}/{\delta})^{O(1)}$. The algorithm GEO-DP parametrized with $k\ge(1/\eps)(1/\delta)^{O(1)}$ tries to subdivide the input
square into the faces $\F(\L)$ and then recurses on the subproblems
given by the faces. Observe that each face in $\F(\L) \setminus \F_+(\L)$
does not contain any rectangle from $\R$ and thus we ignore those faces from now on. We distinguish two types of faces in $\F_+(\L)$:
faces which are homeomorphic to a straight line, and those homeomorphic to a cycle. Note that due to Lemma~\ref{lem:corridor-shapes} and Lemma~\ref{lem:corridor-at-boundaries} no more complex shapes can arise.

Let $F\in\F_+(\L)$ be homeomorphic to a straight line. We claim that GEO-DP finds an optimal solution
for $F$. To get some intuition, let us pretend that $F$ is the union of
a set of complete grid cells and that all rectangles inside $F$ are
blocks, i.e., $g_{i}=1$ or $h_{i}=1$ for each $R_{i}\in\R$ with
$R_{i}\subseteq F$. Then there exists a cut through $F$ which splits
$F$ into two sub-faces without intersecting any rectangle (see Figure~\ref{fig:slice-corridors-cycles}a).
Moreover, the boundary of each sub-face will not be more complex than
the boundary of $F$ itself. Due to this, the complexity
of the subproblems does not increase during the recursion process, and the algorithm GEO-DP finds an optimal solution
for $F$. While for arbitrary faces homeomorphic to a straight line and arbitrary rectangles instead of blocks the analysis is more technical,
and in particular requires circumventing rectangles within $F$, the
key concept is the same.

\begin{lem}\label{lem:dp-for-paths}
Consider a face $F\in\F_+(\L)$ which is homeomorphic to a
straight line. Then GEO-DP parametrized by a value $k\ge(1/\eps)(1/\delta)^{O(1)}$
computes an optimal solution for the DP-cell corresponding to $F$.
\end{lem}

\begin{figure}
\begin{centering}
\begin{tabular}{ccc}
\includegraphics[scale=0.4]{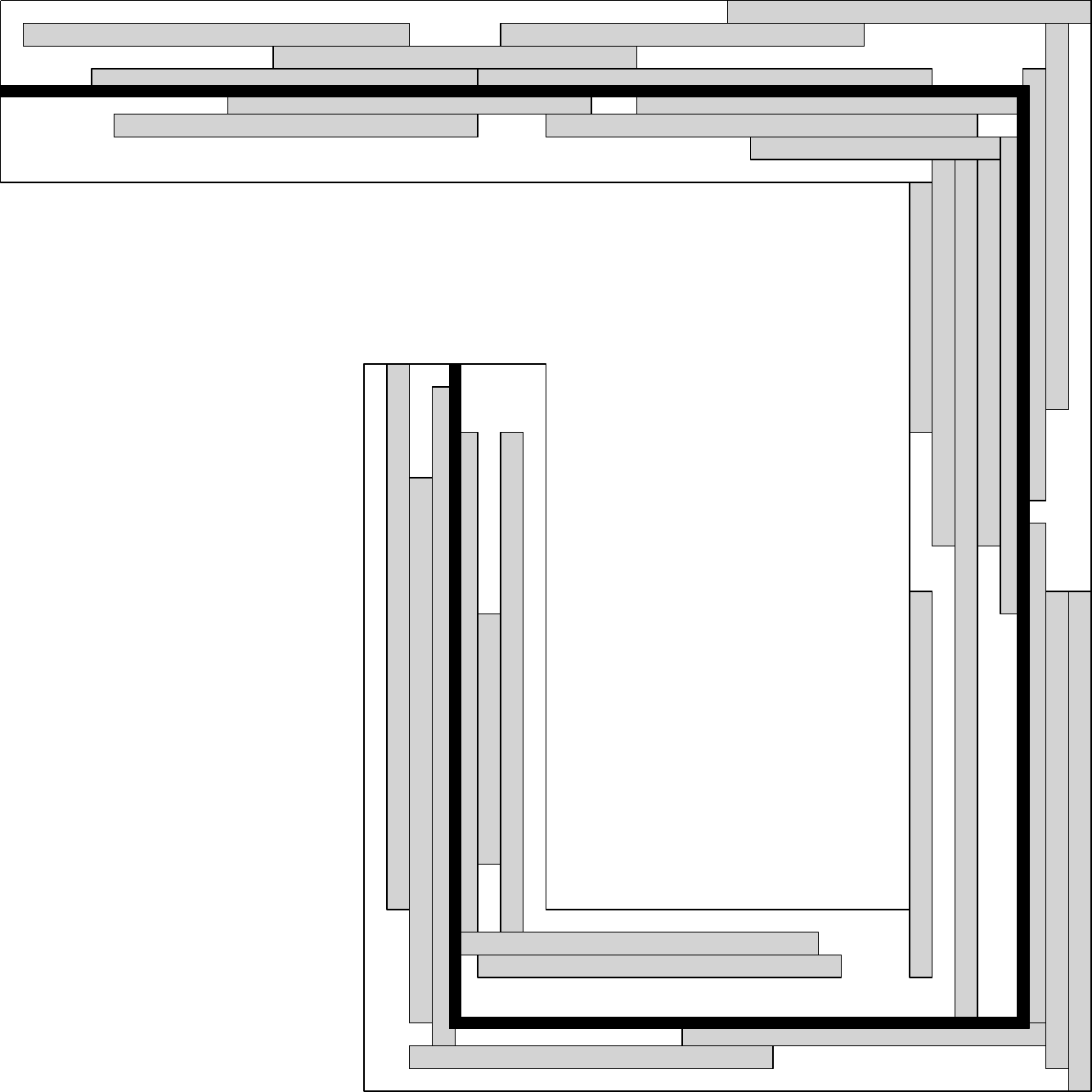} & ~~~~~~~~~~~~~~~~~~ & \includegraphics[scale=0.4]{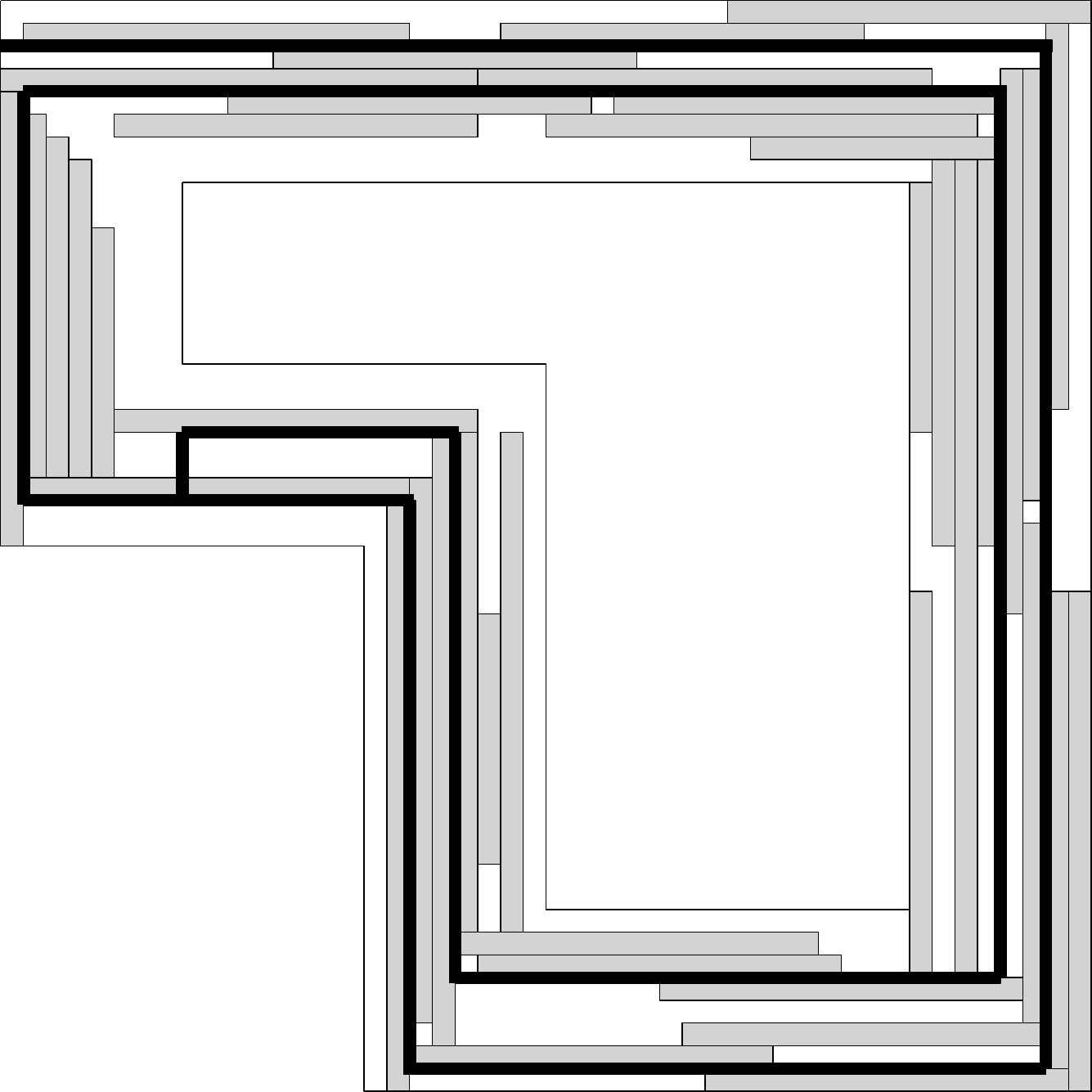} \tabularnewline
(a) &  & (b)\tabularnewline
\end{tabular} 
\par\end{centering}

\caption{\label{fig:slice-corridors-cycles}The thick lines show a cut of a face (a) into two paths with the same
complexity as the original path and (b) into a path and a cycle with
the same complexity as the original cycle.}
\end{figure}

Now consider a face $F\in\F_+(\L)$ which forms a cycle, i.e., which
is homeomorphic to $S^{1}$. Let us pretend again that $F$ is the union
of some complete grid cells and all rectangles in $F$ are blocks.
Then we can split $F$ into a path-face $F_{1}$ and a smaller cycle
$F_{2}$ while 
ensuring that the boundary of the faces $F_{1}$ and $F_{2}$ consists of at most $(1/\eps)(1/\delta)^{O(1)}$ edges each 
(see Figure~\ref{fig:slice-corridors-cycles}b). The recursion terminates when at some recursion level $F_2 = \emptyset$.  
 When doing this operation repeatedly, we ensure that the total weight of intersected rectangles is only an $\eps$-fraction of the total weight of the rectangles in the paths that we detached from the cycle.

Using this construction we can show that
GEO-DP parametrized by sufficiently large $k$ computes a $(1+\eps)$-approximation for $F$, using
that it solves the subproblems for path-faces optimally. Again, for arbitrary
rectangles and more general cycle-faces~$F$ the reasoning is more technical while the core idea stays the same.

\begin{lem}\label{lem:dp-for-cycles}
Consider a face $F\in\F_+(\L)$ which is homeomorphic to $S^{1}$.
Then GEO-DP parametrized by a value $k\ge\frac{1}{\eps}(\frac{1}{\delta})^{O(1)}$
computes a $(1+\eps)$-approximative solution for the DP-cell corresponding
to~$F$.
\end{lem}

When constructing the partition given by the lines $\L$ we intersect (and thus lose) rectangles of
total weight at most $\eps \cdot w(\R)$. When solving the
subproblems given by the faces of the partition we again lose rectangles of total
weight at most $\eps \cdot w(\R)$. Thus, by choosing $k:=(1/\eps)(1/\delta)^{O(1)}$, 
GEO-DP yields a PTAS.

\begin{thm}\label{thm:ptas-for-large}
Let $\eps>0$ and $\delta>0$ be constants. 
\end{thm}

Using standard shifting technique arguments we obtain the following corollary.

\begin{cor}\label{cor:PTAS-roughly-same-size}
Let $\eps>0$ and $\delta>0$ be constants.
The algorithm GEO-DP parametrized by \mbox{$k=(\frac{1}{\eps\cdot \delta})^{O(1)}$} is a polynomial time $(1+\eps)$-approximation algorithm for instances of MWISR where for all rectangles $R_i, R_{i'}$ it holds that
$\max\{h_i,g_i\}\le (1/\delta) \cdot \max\{h_{i'},g_{i'}\}$.
\end{cor}

\newpage{}

\bibliographystyle{abbrv}
\bibliography{MWISR-citations}
\newpage{}


\appendix



\section{Proof of Lemma \ref{lem:good-cut-suffices}}
First, we prove a technical lemma which
shows a polynomial upper bound on the number of connected components, and the complexity of each connected component, of an intersection of a collection of polygons. The lemma will allow us later to bound the complexity of the generated subproblems during the recursion of the DP. 

\begin{lem}
\label{lem:polygons-complexity}
Let $P_1, \ldots, P_m$ be a collection of axis-parallel polygons, where each $P_i$ has at most $\bar{\ell}$ edges. Then the intersection $\bigcap_{i=1,\ldots,m} P_i$ consists of at most $(m\bar{\ell})^2$ connected components, and each connected component $P$ is a polygon with at most $(m\bar{\ell})^2$ edges.
\end{lem}

\begin{proof}
As the polygons $P_i$ have at most $m\bar{\ell}$ edges in total, and any two edges can cross at most once, we have at most $(m\bar{\ell})^2$ pairs of crossing edges. Each connected component of the intersection $\bigcap_{i=1,\ldots,m} P_i$ has at least four corners, and each corner corresponds to a different pair of crossing edges. Hence, there are at most $(m\bar{\ell})^2$ connected components. The number of edges of one connected component $P$ equals the number of vertices of $P$, which is also upper bounded by the number of crossing pairs $(m\bar{\ell})^2$. 
\end{proof}

The following lemma implicitly describes a family of transitions for
the DP-subproblems which we will use in the proof of Lemma \ref{lem:good-cut-suffices}.

\begin{lem}
\label{lem:families-Pj}Let $\alpha, \ell$ be values such that conditions of Lemma \ref{lem:good-cut-suffices} are satisfied, and let $\bar{\R}$ be a set 
of at most $n$ pairwise non-overlapping rectangles with integer coordinates in \{0,\ldots,2n-1\}. Let $j^*=\left\lceil \log_{3/2}n^{2}/\eps\right\rceil$.
Then for each $j \in \left\{ 0,1,\ldots, j^* \right\}$
there is a family of axis-parallel polygons $\P_{j}$ with integer coordinates such that:
\begin{enumerate}[a) ]
\item each polygon $P \in \P_j$ has at most $(j+1)^2 \cdot (\ell+4)^{2}$ edges,
\item $\P_0 = \{[0,2n-1]\times[0,2n-1]\}$,
\item the polygons in $\P_j$ are disjoint, and each polygon $P \in \P_{j-1}$ is a disjoint union of at most $(j+1)^2 (\ell+4)^2 $ polygons from $\P_{j}$,
\item each polygon $P \in \P_{j^{*}}$ contains at most one rectangle from $\bar{R}$,
\item for each set $\P_j$ we get $\sum_{P \in \P_j} w(P)\ge\left(1-\alpha\right)^{j}\cdot w(\bar{\R})$.
\end{enumerate}
\end{lem}

\begin{proof}
We set $\P_0 = \{[0,2n-1]\times[0,2n-1]\}$, i.e., $\P_0$ consists of one rectangle which contains all the rectangles from $\bar{\R}$. We then construct the sets $\P_1,\ldots,\P_{j^*}$ one by one, as follows. To construct $\P_j$, we consider each polygon $P \in \P_{j-1}$, and we add to $\P_j$ the following set of polygons, which together give a disjoint union of $P$.
If $P$ contains at most one rectangle $R \in \bar{\R}$, we add $P$ to the set $\P_{j}$.
Otherwise, if there is a rectangle $R_0  \in \bar{\R}$, $R_0 \subseteq P$ with $w(R_0)\ge\frac{1}{3}\cdot \sum_{R \in \bar{\R}, R \subseteq P} w(R)$, we add to $P_{j}$ the following
polygons: $P \cap R_0 = R_0$, and the connected components of $P \cap \overline{R_0}$, i.e. the connected components of $P\setminus R_0$.
Finally, consider the case that no rectangle $R_0  \in \bar{\R}$ with $R_0 \subseteq P$ has weight $w(R_0)\ge\frac{1}{3}\cdot \sum_{R \in \bar{\R}, R \subseteq P} w(R)$.
Then there exists a $\alpha$-cheap $\ell$-cut for the set of rectangles from $\bar{R}$ which are contained in $P$. Let $P_c$ be the polygon defining this cut, and let $\bar{P_c}$ be its complement
intersected with the input square. We can assume that all corners of $P_c$ have integer coordinates. We add to $\P_{j}$ each connected component of $P\cap P_c$ and $P\cap\bar{P_c}$. Notice that $P_c$ has at most $\ell$ edges, and so $\bar{P_c}$ has at most $\ell+4$ edges.

We now have to check that all required properties are satisfied. 

$a)$ The only polygon in $\P_0$ has $4$ edges. 
Each polygon $P \in \P_j$ for $j \ge 1$, is a connected component of an intersection of at most $j + 1$ polygons with axis-parallel edges and integer corner coordinates, and at most $\ell+4$ edges each. $P$ has axis-parallel edges and integer coordinates, and from Lemma \ref{lem:polygons-complexity} it has at most $(j + 1)^2 (\ell+4)^2$ edges.

$b)$ Defined at the beginning of the proof.

$c)$ From the construction of the sets $\P_j$ it can be easily observed that the polygons in $\P_j$ are disjoint, and each polygon $P \in \P_{j-1}$ is a union of polygons from $\P_{j}$. We now have to upper bound the number of polygons from $\P_{j}$ which can be contained in one polygon $P \in \P_{j-1}$. Polygon $P$ is a connected component of an intersection of at most $j$ polygons, each with at most $\ell+4$ edges. By construction each polygon $P' \in \P_{j}$ is contained in some polygon $P \in \P_{j-1}$. Each polygon $P' \in \P_{j}$ contained in $P$ is a connected component of an 
intersection of $P$ with a polygon $P_c$ with at most $\ell$ edges, or with a polygon $\bar{P_c}$ with at most $\ell+4$ edges. Therefore $P'$ is a connected component of an intersection of at most $j+1$ polygons, each with at most $\ell+4$ edges. From Lemma \ref{lem:polygons-complexity} the number of such components is upper bounded by $(j + 1)^2 (\ell+4)^2$.

$d)$ For a polygon $P$ let $w(P):=\sum_{R \in \bar{\R}: R \subseteq P}w(R)$ denote the weight of all rectangles from $\bar{\R}$ contained in $P$. We will show by induction that if $P \in P_j$ contains more than one rectangle from $\bar{R}$, then $w(P) \le \frac{2}{3}^j w(\bar{\R}) \le \frac{2}{3}^j n^2/\eps$. That value is at most $1$ for $P \in \P_{j^*}$, and as each rectangle from $\bar{\R}$ has weight at least $1$, $P$ cannot contain more than one rectangle.

For the only polygon $P \in \P_0$ we have $w(P) = w(\bar{R}) \le n^2/\eps$. We assume by induction that the property holds for $\P_{j-1}$, and we will show that it holds also for $\P_j$. Let $P \in \P_j$ be contained in a polygon $P_0 \in \P_{j-1}$, where $w(P_0) \le \frac{2}{3}^{j-1} w(\bar{\R}) \le \frac{2}{3}^{j-1} n^2/\eps$. If $P$ contains more than one rectangle from $\bar{\R}$, then either $P \subseteq P_0 \setminus R$ for a heavy rectangle $R$, or $P$ is obtained from $P_0$ by a balanced cut. In both cases we get $w(P) \le \frac{2}{3}w(P_0)$ and we are done.

$e)$ The property holds for $\P_0$, as $\sum_{P \in \P_0} w(P) = w(\bar{\R})$. We will give a proof by induction. Assume that the property holds for $\P_{j-1}$, i.e., $\sum_{P \in \P_{j-1}} w(P)\ge\left(1-\alpha\right)^{j-1}\cdot w(\bar{\R})$. The rectangles which are intersected by $\P_j$, but not by $\P_{j-1}$, must be intersected by the newly introduced polygons $P_c$, which intersect polygons $P \in \P_{j-1}$. As each polygon $P_c$  is a $\alpha$-cheap $\ell$-cut for the set of rectangles contained in the corresponding polygon $P$, we get $\sum_{P' \in \P_j: P' \subseteq P} w(P') \ge (1-\alpha) w(P)$ for each $P \in P_{j-1}$, and so $\sum_{P \in \P_{j}} w(P)\ge\left(1-\alpha\right)^{j}\cdot w(\bar{\R})$.
\end{proof}

With this preparation we are able to prove Lemma~\ref{lem:good-cut-suffices}.

\begin{proof}[Proof of Lemma~\ref{lem:good-cut-suffices}]
Suppose we parametrize GEO-DP by $k:=\left(\left\lceil \log_{3/2}n^{2}/\eps\right\rceil+1\right)^2 \cdot(\ell+4)^{2}$.
Denote by $\P_{j}$, with $j \in \left\{0,1,\ldots,j^*\right\}$ for $j^*=\left\lceil \log_{3/2}n^{2}/\eps\right\rceil$,
the families of axis-parallel polygons with integer coordinates as given in Lemma~\ref{lem:families-Pj}.

From Lemma~\ref{lem:families-Pj}a) any polygon $P\in\P_{j}$ has at most $k$ edges, and so $\P_j \subseteq \P$ and GEO-DP has a DP-cell for $P$.
If $P \in \P_{j^*}$, from Lemma~\ref{lem:families-Pj}d) we know that $P$ contains at most one rectangle, and  so $w(sol(P))=w(P)$ where for each polygon $P$ we denote by $w(P)$ the total weight of all rectangles in $\R$ which are contained in $P$. From Lemma \ref{lem:families-Pj}c) each polygon $P \in \P_j$ is a union of at most $k$ polygons $P_1,\ldots, P_m \in \P_{j+1}$. Therefore GEO-DP tries the subdivision of $P$ into these components and we get that $w(sol(P)) \ge \sum_{i=1}^m w(sol(P_i))$, which for the input polygon $P_0 \in \P_0$ (see Lemma~\ref{lem:families-Pj}b) gives 
$$w(sol(P_0)) \ge \sum_{P\in\P_{j^*}}w(sol(P)) =  \sum_{P\in\P_{j^*}}w(P) \ge\left(1-\alpha\right)^{j^*}\cdot w(\bar{\R})\enspace,$$
where the last inequality comes from Lemma~\ref{lem:families-Pj}e).
Therefore, the overall approximation ratio of GEO-DP is $\left(1-\alpha\right)^{-j^{*}}=\left(\frac{1}{1-\alpha}\right)^{\left\lceil \log_{3/2}n^{2}/\eps\right\rceil } = (1+\alpha)^{O(\log(n/\eps))}$ 
when parametrized by $k=\left(\left\lceil \log_{3/2}n^{2}/\eps\right\rceil+1\right)^2 \cdot(\ell+4)^{2}$. 
\end{proof}


\section{Proofs from Section \ref{sec:QPTAS}}\label{apx:QPTAS}

\begin{defn}
Let $\R=\{R_{1},...,R_{n}\}$ and $\bar{\R}=\{\bar{R}_{1},...,\bar{R}_{n}\}$ be sets of rectangles s.t. $R_i$ and $\bar{R}_i$ have coordinates 
$x_{i}^{(1)},y_{i}^{(1)},x_{i}^{(2)},y_{i}^{(2)}$ and $\bar{x}{}_{i}^{(1)},\bar{y}{}_{i}^{(1)},\bar{x}{}_{i}^{(2)},\bar{y}{}_{i}^{(2)}$, 
respectively. We say that $\R$ and $\bar{\R}$ are \emph{combinatorially
equivalent} (or \emph{equivalent }for short) if we
have that $w(R_{i})=w(\bar{R}_{i})$, $x_{i}^{(t)}\le x_{i'}^{(t')}\Leftrightarrow\bar{x}_{i}^{(t)}\le\bar{x}_{i'}^{(t')}$,
$x_{i}^{(t)}<x_{i'}^{(t')}\Leftrightarrow\bar{x}_{i}^{(t)}<\bar{x}_{i'}^{(t')}$,
$y_{i}^{(t)}\le y_{i'}^{(t')}\Leftrightarrow\bar{y}_{i}^{(t)}\le\bar{y}_{i'}^{(t')}$,
and $y_{i}^{(t)}<y_{i'}^{(t')}\Leftrightarrow\bar{y}_{i}^{(t)}<\bar{y}_{i'}^{(t')}$,
for all $t,t'\in\{1,2\}$ and all $i,i'\in\{1,...,n\}$.
\end{defn}

\begin{proof}[Proof of Lemma~\ref{lem:weighted-stretching}]
W.l.o.g. we can assume that $x_{i}^{(1)},y_{i}^{(1)},x_{i}^{(2)},y_{i}^{(2)} \in \{0,\ldots,2n-1\}$ for each $R_i \in \R$, where $n = |\R|$. The set of rectangles $\bar{\R}$ will consist of rectangles $\bar{R}_{1},...,\bar{R}_{n}$, where for each $i = 1,\ldots,n$ we have $w(\bar{R}_{i}) = w(R_i)$. We set the coordinates $\bar{x}{}_{i}^{(1)},\bar{y}{}_{i}^{(1)},\bar{x}{}_{i}^{(2)},\bar{y}{}_{i}^{(2)}$ of $\bar{R}_{i}$ as follows.

For $j = 1,\ldots,2n$ we define $\R_x(j):=\{R_i \in \R: x_{i}^{(1)}<j \}$ and $\R_y(j):=\{R_i \in \R: y_{i}^{(1)} < j\}$. For $i \in\{1,...,n\}$ and $t\in\{1,2\}$ we set 
$$\bar{x}{}_{i}^{(j)}:=x_i^{(j)}+\left\lceil w(\R_x(x_i^{(j)})) \cdot \frac{2|\R|}{w(\R)} \right\rceil, \bar{y}{}_{i}^{(j)}:=y_i^{(j)}+\left\lceil w(\R_y(y_i^{(j)}))\cdot \frac{2|\R|}{w(\R)} \right\rceil \enspace.$$

As the weights $w(\R_x(j))$ and $w(\R_y(j))$ are monotonically non-decreasing with $j$, the sets $\R$ and $\bar{\R}$ are equivalent, and in particular for any $t,t'\in\{1,2\}$ and $i,i'\in\{1,...,n\}$ we have $x_{i}^{(t)}<x_{i'}^{(t')}\Leftrightarrow\bar{x}_{i}^{(t)}<\bar{x}_{i'}^{(t')}$.

As for any $j$ we have $w(\R_x(j)), w(\R_y(j)) \in [0,w(\R)]$, the rectangles from $\bar{\R}$ have integer coordinates in $\{0,...,4 \cdot |\R| \}$. We now have to show that $\bar{\R}$ is well-distributed. 

W.l.o.g. it is enough to show that for any $\gamma > 0$ and any vertical stripe $S$ of the square $[0,4|\R|]\times[0,4|\R|]$ of width $\gamma \cdot 4|\R|$  all rectangles from the set $\bar{\R}$ contained in $S$
have a total weight of at most $2 \gamma \cdot w(\bar{\R})$.

Let $\bar{\R}(S)$ be the set of rectangles from $\bar{\R}$ contained in $S$, and assume that  $\bar{\R}(S) \neq \emptyset$. Let $\bar{R}_{\ell}$ and $\bar{R}_r$ be rectangles from $\bar{\R}(S)$ minimizing $\bar{x}{}_{i}^{(1)}$ and maximizing $\bar{x}{}_{i}^{(2)}$, respectively. We have $\bar{x}{}_{r}^{(2)} - \bar{x}{}_{\ell}^{(1)} \le 4 \gamma |\R|$.  As $\bar{x}{}_{r}^{(2)} > \bar{x}{}_{\ell}^{(1)}$, we have $x_{r}^{(2)} > x_{\ell}^{(1)}$, and:
$$\bar{x}{}_{r}^{(2)} - \bar{x}{}_{\ell}^{(1)}
= (x_{r}^{(2)} - x_{\ell}^{(1)})+\left( \left\lceil w(\R_x(x_r^{(2)})) \cdot \frac{2|\R|}{w(\R)} \right\rceil - \left\lceil w(\R_x(x_{\ell}^{(1)})) \cdot \frac{2|\R|}{w(\R)} \right\rceil \right)$$
$$\ge \left( w(\R_x(x_r^{(2)})) - w(\R_x(x_{\ell}^{(1)})) \right) \cdot \frac{2|\R|}{w(\R)} 
\ge w(\bar{\R}(S)) \cdot \frac{2|\R|}{w(\R)} \enspace,$$
as $\bar{\R}(S) \subseteq \R_x(x_r^{(2)}) \setminus \R_x(x_{\ell}^{(1)})$.
We get $4 \gamma |\R| \ge w(\bar{\R}(S)) \cdot \frac{2|\R|}{w(\R)}$, which gives us $w(\bar{\R}(S)) \le 2 \gamma w(\R) = 2 \gamma w(\bar{\R})$. The set of rectangles $\bar{\R}$ is well-distributed.
\end{proof}

\begin{proof}[Proof of Lemma~\ref{lem:rectangle-faces}]
There are $(1/\delta)^2-1$ vertical grid lines which can cut rectangles from the set $\R$. Each of the grid lines has length $N$, so it cuts less than $(1/\delta)^2$ large vertical rectangles. An analogous condition holds for horizontal grid lines and large horizontal rectangles, giving an upper bound of $2(1/\delta)^4$ on the number of rectangle faces.
\end{proof}

\begin{proof}[Proof of Lemma~\ref{lem:G-some-properites}]
Take the embedding of $G(\L)$ which is induced by the lines $\L$. 
By construction of the set  $\L$, the lines in  $\L$ do not intersect properly. Thus, due to the definition of $G(\L)$ this yields a planar embedding of $G(\L)$.

Now we bound the number of vertices and edges of $G(\L)$. 
Each vertex of $G(\L)$ is an endpoint of a line from $\L$. From Lemma \ref{lem:rectangle-faces} there are at most $2(1/\delta)^4$ rectangle faces, which yield at most $8(1/\delta)^4$ vertices of $G(\L)$. All remaining vertices are endpoints of lines contained in single grid cells. As each grid cell $Q$ can be intersected by at most $4$ rectangle faces, and by at most two lines of $\L$ corresponding to rectangles crossing $Q$, that gives at most $12$ new vertices per each grid cell. As the number of grid cells is $(1/\delta)^4$, we get $|V| \le 20(1/\delta)^4$.

As $G(\L)$ is planar, and all edges of $G(\L)$ are horizontal or vertical, the degree of each vertex is at most $4$ and we get $|E| \le 40(1/\delta)^4$.
\end{proof}

\begin{proof}[Proof of Lemma~\ref{lem:rec-intersect-constant}]
The only lines from $\L$ which intersect rectangles from $\R$ are the lines which lie on the boundary of the grid cells, as all other lines lie on the boundaries of some rectangles from $\R$, and the rectangles in $\R$ are pairwise non-overlapping. The only vertices of $G(\L)$ which can lie in the interior of any rectangle from $\R$ are the corners of the grid cells, as all remaining vertices of $G(\L)$ lie on the boundaries of rectangles from $\R$ (either a rectangle generating a rectangle face, or a rectangle crossing a grid cell). Therefore, if a rectangle $R \in \R$ intersects at most $m$ grid cell boundaries, it is intersected by at most $m$ edges of $G(\L)$. We instantly get that a rectangle from $\R_S$ is intersected by at most $4$ edges of $G(\L)$.

Let $R \in \R_L$ be a rectangle contained in a single row or column of grid cells. Let $Q$ and $Q'$ be the extremal grid cells intersected by $R$ (i.e., such that $R$ intersects $Q$ and $Q'$, and the shorter edges of $R$ are contained in $Q$ and $Q'$). We will show that $R$ can be intersected by edges of $G(\L)$ only at the boundaries of $Q$ and $Q'$, i.e., $R$ is intersected by at most $2$ edges of $G(\L)$. Consider a grid cell boundary $e = Q_1 \cap Q_2$ for some grid cells $Q_1$ and $Q_2$, where $R$ crosses $Q_1$ and $Q_2$. Then $R \in \R_{Q_1}, \R_{Q_2}$, and the lines added to $\L$ while considering the grid cells $Q_1$ and $Q_2$ do not intersect $R$.

Let us consider the last case. Let $R \in \R_L$ be a rectangle which is not contained in a single row or column of grid cells. Then $R$ is contained in a rectangle face, and it is not intersected by any edges of $G(\L)$.
\end{proof}

\begin{proof}[Proof of Lemma~\ref{lem:rec-intersect-faces}]
From the construction of the lines $\L$ we obtain the following propositions. 

\begin{prop}\label{prop:maze-grid-corners}
Let $p$ be a corner of a grid cell. If $p$ does not lie on any line $L \in \L$, then $p \in F$ for some rectangle face $F \in \F(\L)$.
\end{prop}

\begin{prop}\label{prop:maze-grid-bends}
Let $e$ be a horizontal (resp. vertical) edge of a grid cell $Q$, and let $p \in e$ such that $p$ is not a corner of $Q$. If $p$ does not lie on a rectangle face, and $p$ does not lie on a line from $\L$, then $Q$ is crossed by a large vertical (resp. horizontal) rectangle.
\end{prop}

Let $F \in \F(\L)$ be a face contained in a single grid cell $Q$. As all large rectangles not contained in the rectangle faces are contained in a single row or column of grid cells, and all small rectangles have width and height at most $\delta^2 N$, all rectangles which have non-empty intersection with $F$ are contained in a horizontal stripe of $I$ of width $3 \delta^2 N$, or in a vertical stripe of $I$ of width $\delta^2 N$ (see Figure \ref{fig:stripes}a). As the set of rectangles $\R$ is well-distributed, we get that the total weight of rectangles intersecting $F$ is at most $8 \delta^2 w(\R)$. 

\begin{figure}
\begin{centering}

\centerline{%
\includegraphics[height=0.4\textwidth]{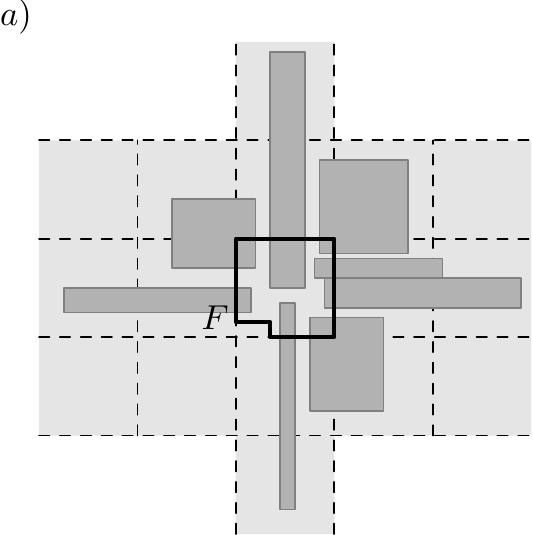}%
\hskip0.2\textwidth
\includegraphics[height=0.4\textwidth]{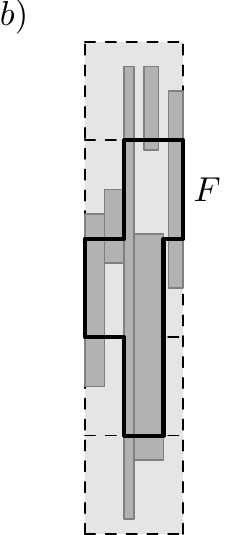}%
}%
\par\end{centering}
\caption{\label{fig:stripes}All rectangles from $\R$ intersecting a face $F \in \F(\L)$ which is not a rectangle face are contained in thin stripes of the input square. The rectangles of $\R$ intersecting $F$ are depicted in gray. The shaded area denotes the stripes.}
\end{figure}

Let $F \in \F(\L)$ be a face which is not a rectangle face, and which is not contained in a single grid cell. We will show that $F$ is contained in a single row or column of grid cells. Assume, for contradiction, that $F$ is not contained in a single row or column of grid cells. Then there must be a grid cell $Q$ with vertical and horizontal edges $e$ and $e'$, respectively, such that $F \cap e \neq \emptyset$ and $F \cap e' \neq \emptyset$. As $F \cap \bigcup_{L \in \L} L = \emptyset$, and therefore $F \cap E = \emptyset$, Proposition  \ref{prop:maze-grid-corners} and Proposition \ref{prop:maze-grid-bends} give us that $Q$ is crossed both by horizontal and vertical rectangles from $\R$, which gives contradiction. The face $F$ must be contained in a single row or column of grid cells.
 
Assume w.l.o.g. that $F$ is contained in a single column of grid cells, but not in a single grid cell (see Figure \ref{fig:stripes}b). Then for each grid cell $Q$ for which $F \cap Q \neq \emptyset$, $F \cap Q$ is contained between the lines $L_{\ell}$ and $L_r$ which are parts of edges of vertical rectangles crossing $Q$. In particular, no rectangle from $\R$ can intersect $L_{\ell}$ and $L_r$. If a rectangle from $\R$ has non-empty intersection with $F \cap Q$, then it must be contained in the same column of grid cells as $F$, i.e., all rectangles intersecting $F$ are contained in a stripe of $I$ of width $\delta^2 N$, and have total weight at most $2 \delta^2 w(\R)$. 

As the boundary of a rectangle face does not intersect any rectangles from $\R$, we instantly get that the boundary of any face $F \in \F(\L)$ intersects rectangles from $\R$ of total weight at most $8 \delta^2 w(\R)$. 
\end{proof}

\begin{proof}[Proof of Lemma~\ref{lem:G-cost-weight}]
From Lemma \ref{lem:rec-intersect-constant} each rectangle from $\R$ can be intersected by at most $4$ edges of the graph $G(\L)$. That gives us that the total cost of edges in $G(\L)$ is at most $4 w(\R)$.

From Lemma \ref{lem:rec-intersect-faces} each face of $F$ which is not a rectangle face has non-empty intersection with rectangles from $\R$ of total weight at most $8 \delta^2 w(\R)$, and so the weight of $F$ is at most $8 \delta^2 w(\R)$.

As each rectangle $R \in \R$ intersecting $m$ faces contributes $w(R)/m$ to the weight of each of the $m$ faces, the total weight of the faces is  $w(\R)$.
\end{proof}

\begin{proof}[Proof of Lemma~\ref{lem:balanced-cut-parametrized}]
Let $\L$ be the set of lines, and $G(\L)=(V,E)$ the graph constructed for the set of rectangles $\R$. From Lemma \ref{lem:G-some-properites} $G(\L)$ is planar, and so we can apply Theorem \ref{thm:cycle-separator} to the embedding given by the lines $\L$.

Let $C$ be the V-cycle of $G(\L)$ from Theorem \ref{thm:cycle-separator} for $\bar{k}=1/\delta$. We will transform $C$ into a cycle $C'$ which uses only ordinary edges of $G(\L)$. We consider the face edges one by one, and we substitute each face edge $uv$ for a face $F$ with a path in $G(\L)$ connecting $u$ and $v$ and using only edges which are on the boundary of $F$. We can choose this path in two ways, depending on whether we want $F$ to become a part of the interior, or the exterior of $C'$. We always merge $F$ with the part of lower weight. Notice that $C'$ might not be a simple cycle, but we can always modify $C'$ so that each edge appears only $O(1)$ times in $C'$.

We will show that the cycle $C'$ gives a balanced $O(\delta)$-cheap $O((1 / \delta)^4)$-cut. From Lemma \ref{lem:G-some-properites} we get that $|E| = O((1 / \delta)^4)$. The cycle $C'$ uses $O((1 / \delta)^4)$ edges, and so $C'$ is a $O((1 / \delta)^4)$-cut.

We will now upper bound the total weight of rectangles from $\R$ intersected by $C'$. From Theorem \ref{thm:cycle-separator} the ordinary edges of $C$ have cost $O(M/\bar{k})$, which from Lemma \ref{lem:G-cost-weight} is $O(\delta w(\R))$, and so they intersect rectangles from $\R$ of a total weight $O(\delta w(\R))$. The remaining edges of $C'$ lie on the boundaries of at most $\bar{k}=1 / \delta$ faces, and from Lemma \ref{lem:rec-intersect-faces} the boundary of each face intersects rectangles of weight $O(\delta^{2}\cdot w(\R))$. The edges of $C'$ intersect rectangles of total weight $O(\delta w(\R))$, and so $C'$ is a $O(\delta)$-cheap cut.

From Theorem \ref{thm:cycle-separator} the interior and the exterior of $C$ have weights at most $2W/3$, and from Lemma \ref{lem:G-cost-weight} we get that $W= w(\R)$. Each rectangle in $\R$ has weight smaller than $w(\R)/3$, and so the weight of each rectangle face of $G(\L)$ is smaller than $w(\R)/3$. From Lemma \ref{lem:G-cost-weight} the weight of any other face of $G(\L)$ is at most $8 \delta^{2}\cdot w(\R)$, which is also smaller than $w(\R)/3$ for $\delta < 1/5$. From the construction of $C'$ the interior and the exterior of $C'$ have weights at most $2w(\R)/3$. The cut $C'$ is balanced. 
\end{proof}

\begin{proof}[Proof of Theorem~\ref{thm:qptas}]
From Lemma \ref{lem:balanced-cut-parametrized}, for any $1/5 > \delta >0$ and for any set $\R$ of pairwise non-overlapping rectangles which does not contain a rectangle of weight at least $w(\R)/3$
there exists a balanced $(c \cdot \delta)$-cheap $O((1 / \delta)^4)$-cut for some constant $c > 0$.

Applying Lemma \ref{lem:good-cut-suffices} gives us, that the algorithm GEO-DP has approximation ratio $(1+c \cdot \delta)^{O(\log(n/\eps))}$ when parametrized by some $k=(1 / \delta)^8\cdot O(\log^2(n/\eps))$. Let us fix $\delta = \Theta\left(\frac{\eps}{\log (n/\eps)}\right)$ such that the approximation ratio is at most $1+\eps$. 
Such choice of $\delta$ requires $k = (\frac{\log n}{\eps})^{O(1)}$.
The running time of GEO-DP is then $n^{(\log n /\eps)^{O(1)}}$ according to Proposition \ref{prop:DP-time}, and so GEO-DP is a QPTAS for the maximum weight independent set of rectangles problem.
\end{proof}


\section{Complete Construction of $\Le$}\label{app:construction-le}

For any two points $p,p'$ we denote by $L[p,p']$ the straight line from $p$ to $p'$. Also, we define $L(p,p):=L[p,p']\setminus \{p,p'\}$.

For each endpoint $p_0$ of a line $L \in \L_{0}$ such that $L$ does not hit a perpendicular line in $\L_{0} \cup \Le$ at $p_0$ we will create a path of lines connecting $L$ with a line in $\L_{0} \cup \Le$, and we will add the constructed lines to the set $\Le$. Ideally, we would like the added lines to intersect no blocks. However, as we want the size of the partition to be small, we will have to allow the lines to cut some blocks.

Before we show the construction of the paths, we need the following lemma. Note that it holds for arbitrary lines $L$, and not only for lines in $\L_{0} \cup \Le$.

\begin{lem}\label{lem:line-hits-block}
Let $Q$ be a grid cell and $p \in Q$. Let $L$ be a line which does not intersect any blocks and lines from $\L_{0}$, which has one endpoint at $p$, and the other endpoint outside of $Q$. If $L$ hits a perpendicular block $B$ at $p$, but it does not hit a perpendicular line from $\L_{0}$ at $p$, then:
\begin{itemize}
\item $p \in \int(Q)$, and
\item one end of $B$ is in $Q$, and the other one is outside of $Q$.
\end{itemize}
\end{lem} 
\begin{proof}
Assume w.l.o.g. that $L$ is vertical, $p$ is at the top end of $L$, and $B$ crosses the boundary of the grid cell to the right of $p$. Let $L'$ be the maximal line which contains the bottom edge of $B$ and does not intersect any blocks or lines from $\L_{0}$. As $p \in L'$, we get that $L' \notin \L_{0}$.

Assume that $p$ is at the boundary of $Q$. If $p$ lies on the bottom or right edge of $Q$, $L'$ is the bottom-most long line crossing the right edge of $Q$. If $p$ lies on the top edge of $Q$, $L'$ is the top-most long line crossing the right edge of $Q$. If $p$ lies on the left edge of $Q$, $L'$ is the bottom-most long line crossing the left edge of $Q$. In each of the cases we have that $L' \in \L_{0}$, which gives a contradiction.

As $p \in \int(Q)$, $B$ intersects $Q$. If $B$ does not end in $Q$, then $L'$ cuts $Q$. If $L'$ is the bottom-most sticking-in line for the left edge of $Q$ then $L' \in \L_{0}$, which gives a contradiction. Otherwise, the bottom-most sticking-in line for the left edge of $Q$ is below $L'$ and cuts $Q$, so it intersects $L$, and again we get a contradiction, as $L$ does not intersect edges from $\L_{0}$. Block $B$ must end in $Q$.
\end{proof}

We now present the construction of the paths. For each endpoint $p_0$ of a line $L \in \L_{0}$ such that $L$ does not hit a perpendicular line in $\L_{0} \cup \Le$ at $p_0$, we construct a set of lines $\{L_1, \ldots , L_m\}$ as follows. Let $Q$ be a grid cell such that $p_0 \in Q$, and let $B$ be the block hit by $L$ at $p_0$. Such a block exists, as $L$ is maximal, and $L$ does not hit a line from $\L_{0} \cup \Le$ at $p_0$. The line $L$ together with $p_0$, $Q$ and $B$ satisfy the conditions of Lemma \ref{lem:line-hits-block}. We get that $p_0 \in \int(Q)$, and $B$ has one end in $Q$. Let $Q' \neq Q$ be the grid cell with the other end of $B$. Let $L_{max}$ be a maximal line which contains the edge of $B$ containing $p$, and does not intersect any blocks or lines from $\L_{0} \cup \Le$. Let $p_1$ be the endpoint of $L_{max}$ such that $L[p_0,p_1] \cap Q' \neq \emptyset$. We set $L_1 = L[p_0,p_1]$. We know that $p_1 \notin Q$, so $L_1$ intersects at least two grid cells.
See Figure~\ref{fig:creating-le-cases} for a sketch.

We fix a parameter $M = 64 \frac{1}{\eps}(\frac{1}{\delta})^{2}$. We have to consider the following cases:

\begin{figure}
\begin{centering}
\centerline{%
\includegraphics[height=0.25\textwidth]{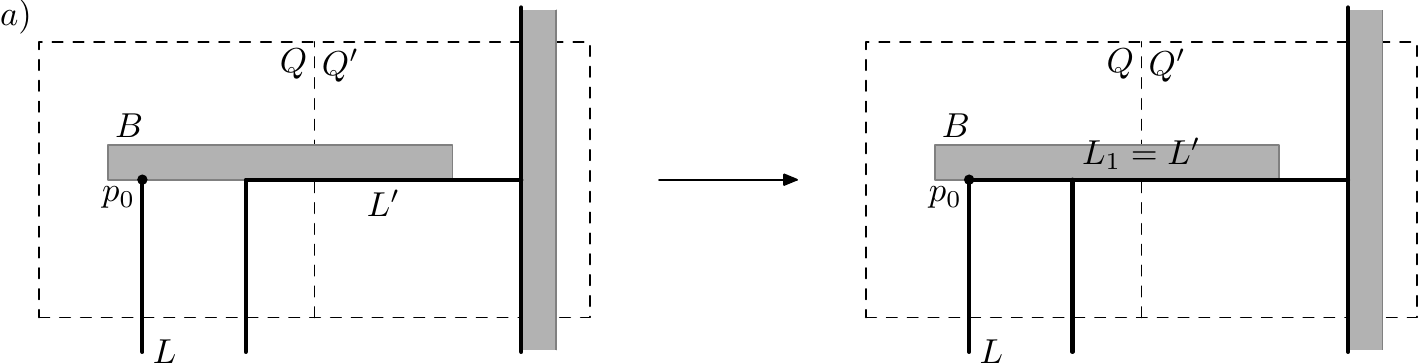}%
}%
\vskip0.03\textwidth
\centerline{%
\includegraphics[height=0.25\textwidth]{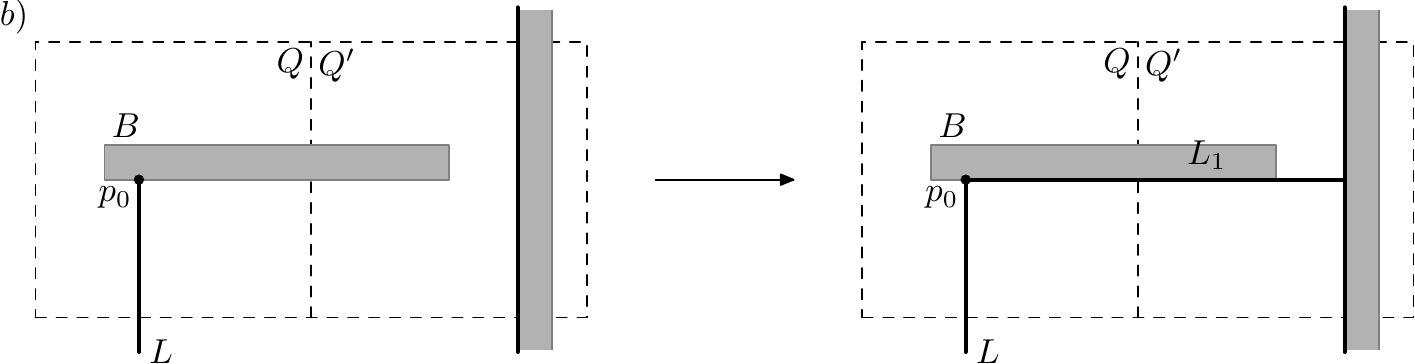}%
}%
\vskip0.03\textwidth
\centerline{%
\includegraphics[height=0.25\textwidth]{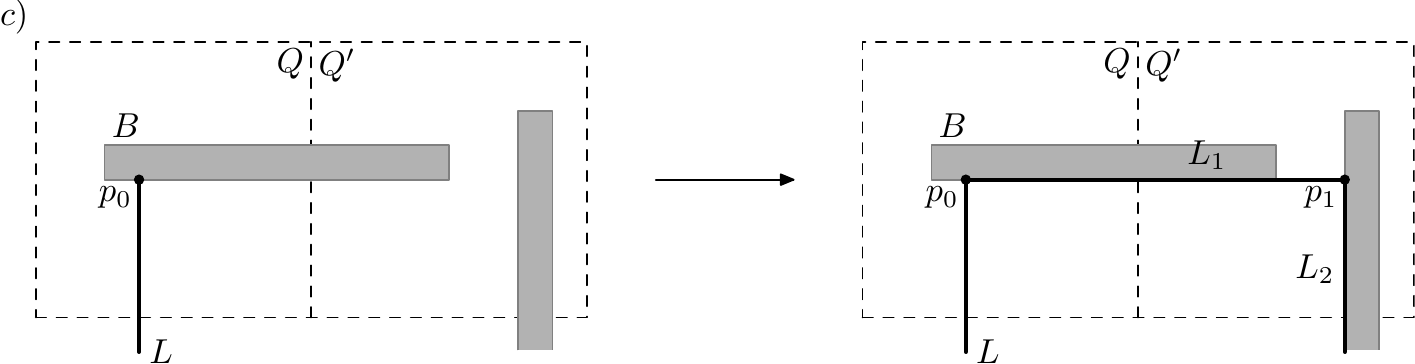}%
}%
\end{centering}
\caption{\label{fig:creating-le-cases}
Construction of the paths. The line $L$ hits the block $B$ at the point $p_0$. We construct a line $L_1$ 
starting at $p_0$, following the bottom edge of $B$ until we hit a perpendicular line which is already in $\L_0 \cup \Le$, or a perpendicular block. In case a) the new line $L_1$ overlaps an existing line $L' \in \L_0 \cup \Le$ in which case we extend $L'$ so that it reaches $p_0$ (and do not add $L_1$ to $\Le$).
In case b) we simply add $L_1$ to $\Le$ and we are done. In case c) we continue constructing the path from the point $p_1$, where $L_1$ hits a perpendicular block.
}
\end{figure}

\begin{enumerate}
\item There is a line $L' \in \Le$ such that $|L_1 \cap L'| > 1$ (see Figure \ref{fig:creating-le-cases}a). 
Such a situation can happen, as lines from $\Le$ are not necessarily maximal. As $L_1$ does not hit $L'$ at $p_0$ and $L_1$ cannot be extended beyond $p_1$, we get that $L' \subseteq L_1$. The construction of the path is finished.

\item Case $1)$ does not happen, but $L_1$ hits a line from $\L_{0} \cup \Le$ at $p_1$, (see Figure \ref{fig:creating-le-cases}b). 
The construction of the path is finished.

\item Cases $1)$ and $2)$ do not happen. In this case, $L_1$ hits some perpendicular block at $p_1$ (see Figure \ref{fig:creating-le-cases}c). 
We proceed as before, considering the line $L_1$ and its endpoint $p_1$ instead of $L$ and $p$. The conditions of Lemma \ref{lem:line-hits-block} are satisfied, as $L_1$ intersects at least two grid cells. We continue extending the path, until one of the cases $1)$ or $2)$  applies, or the number of lines in the path reaches the upper bound $M$.
\end{enumerate}

\begin{figure}
\begin{centering}
\includegraphics[scale=0.75]{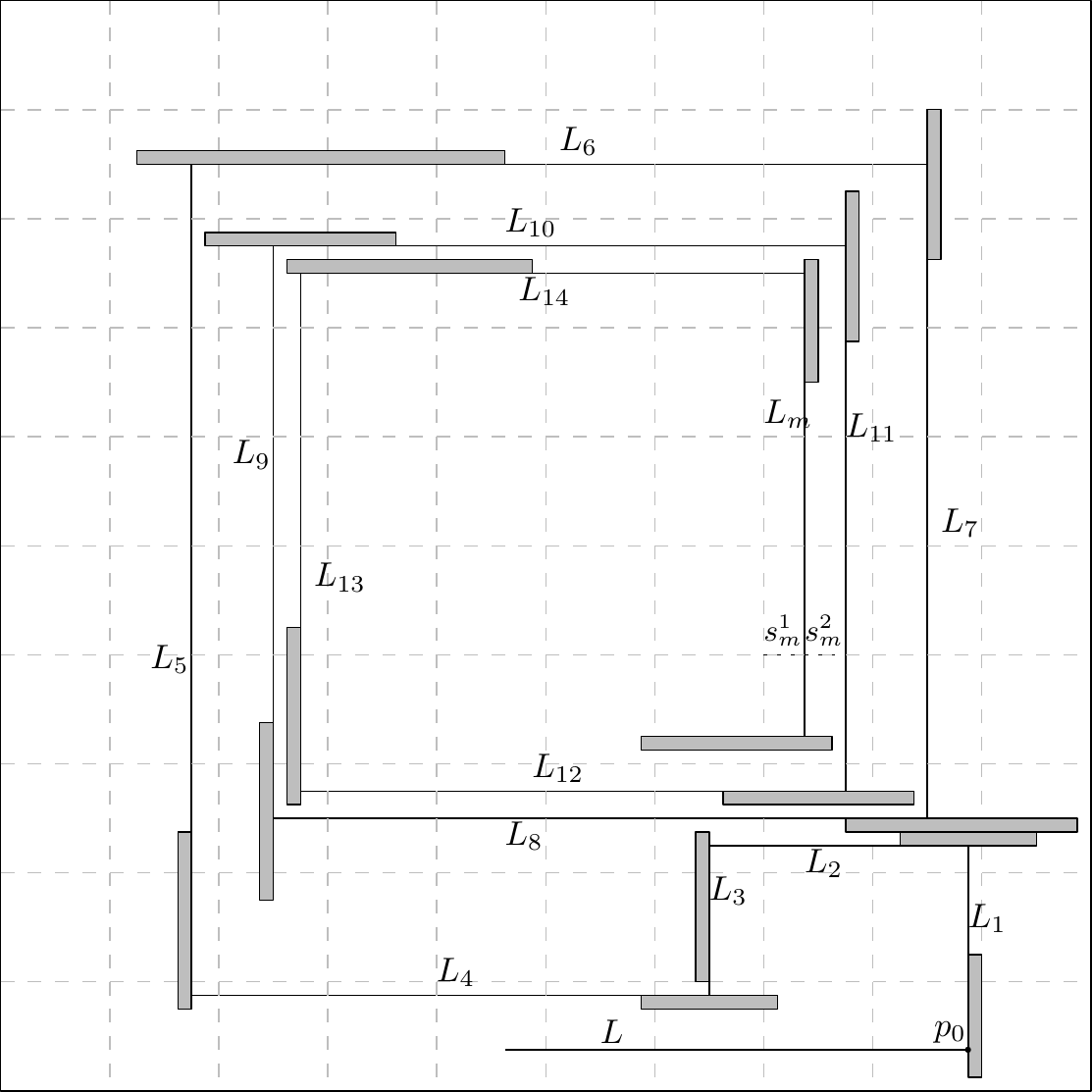}
\par\end{centering}

\caption{\label{fig:construction-lines-L'-case3}The construction of the lines
$\Le$. The blocks of the considered instance are depicted in gray.}
\end{figure}

Let $\{L_1, \ldots, L_m\}$ be the collection of lines obtained as described above, for some $m \le M$. We have $L_i = L[p_{i-1},p_i]$. We modify the set $\Le$ as follows.
If the construction of the set ended in case $1)$, we add the set of lines $\{L_1, \ldots, L_{m-1}\}$ to $\Le$. Let $L' \in \Le$ be the line contained in $L_m$. We extend $L'$, so that it has an endpoint in $p_{m-1}$. Notice, that all the lines which were touching $L'$ after extending are still touching $L'$. 
If the construction of the set ended in case $2)$, we add the set of lines $\{L_1, \ldots, L_m\}$ to $\Le$.

Let us now consider the difficult case, i.e., when after $M$ steps the path $L_1, \ldots , L_M$ does not hit any line from $\L_{0} \cup \Le$. We do not want to extend the path any further, as the number of lines in $\Le$ would become too large. We have to find a place to create a "shortcut"  which connects some part of the path $L_1, \ldots , L_M$ to a line from $\L_{0} \cup \Le$, cutting some blocks. We will ensure that the rectangles cut during this operation have small total weight. The cut will go along the boundary of some grid cell. An example can be seen on Figure \ref{fig:construction-lines-L'-case3}.

Each line $L_i$ intersects least two grid cells. For each $i \in \{1,...,M\}$ let $Q_i$ be the grid cell in which $L_i$ ends, i.e, such that $p_i \in Q_i$. From the construction of the path we know that $p_i \in \int(Q_i)$. Let $e_i$ be the edge of $Q_i$ intersected by $L_i$, and let $p_i^0 = L_i \cap e_i$.  Let $s_i^1$ and $s_i^2$ be the segments on $e_i$ which connect $L_i$ with the two neighboring lines from $\L_{0} \cup \Le \cup \{L_1, \ldots , L_M\}$, i.e. for $j=\{1,2\}$ we have $s_i^j = [p_i^0,p_i^j]$ such that $p_i^j \in \L_{0} \cup \Le \cup \{L_1, \ldots , L_M\}$ and $(p_i^0,p_i^j) \cap (\L_{0} \cup \Le \cup \{L_1, \ldots , L_M\}) = \emptyset$, and $p_i^1 \neq p_i^2$.

We will show that such segments always exist. $L_i$ crosses the edge $e_i$ of $Q_i$, and the maximal line containing $L_i$ which does not intersect any blocks or lines from $\L_{0}$ is long (as at least some part of $L_i$ goes along a long edge of a block). As $L_i \notin \L_{0}$, $L_i$ lies between two lines from $\L_{0}$ --- the leftmost and the rightmost long lines crossing the edge $e_i$ of $Q_i$, and so the segments $s_i^1$ and $s_i^2$ are contained in $e_i$. A segment $s_i^j$ possibly cuts some blocks. Let $\R_i^j$ be the set of rectangles cut by $s_i^j$.

\begin{lem}\label{lem:segments-sij}
Any rectangle $R \in \R_L$ belongs to at most four 
sets $\R_i^j$, where $i \in \{1, \ldots , M\}$ and $j \in \{1,2\}$.
\end{lem}
\begin{proof}
If $R$ is not contained in a single row or column of grid lines, it cannot be cut by any segment $s_i^j$, and so it does not belong to any set $\R_i^j$. Assume w.l.o.g. that $R$ is vertical and it is inside a single column of grid cells. Let $Q_1$ and $Q_2$ be the grid cells where the blocks of $R$ end. 

Let $e$ be a grid cell boundary which intersects $R$ such that $e \cap Q_1 = \emptyset$ and $e \cap Q_2 = \emptyset$. Let $Q$ and $Q'$ be the grid cells for which $e = Q \cap Q'$. $R$ crosses $Q$ and $Q'$, and so the vertical sticking-in lines in $Q$ and $Q'$ cross $Q$ and $Q'$, and $R$ is contained between them. If a line $L_i$ ends in $Q$ or $Q'$, its end cannot lie between the sticking-in lines (as there are no perpendicular blocks which could be hit by $L_i$), and so $R \cap s_i^j = \emptyset$ for $j \in \{1,2\}$, and $R \notin \R_i^j$.

At the boundary $e$ of $Q_1$ or $Q_2$ the rectangle $R$ can be cut by at most two segments $s_i^j$ and $s_{i'}^{j'}$, such that $L(p_i^0, p_{i'}^0) \cap \{L_1, \ldots , L_M\} = \emptyset$, and $(R \cap e) \subseteq L[p_i^0, p_{i'}^0]$.
\end{proof}

From Lemma \ref{lem:segments-sij} we get that there is a line $L_i$ in $\{L_1, \ldots, L_M\}$ and a segment $s_i^j$ for which $w(\R_i^j) \le \frac{2}{M} w(\R_L)$. We have two cases:
\begin{itemize}
\item $s_j^j$ connects the line $L_i$ with a line from $\L_{0} \cup \Le \cup \{L_j\}_{j < i}$. 

We then continue the path from $L$ only until the point $p_i^0 = L_i \cap s_i^j$, and then extend it by the segment $s_i^j$. Formally, we add the lines $L_1, \ldots , L_{i-1}$, the line $L[p_{i-1}, p_i^0] \subseteq L_i$ and the segment $s_i^j$ into $\Le$.

\item  $s_j^j$ connects the line $L_i$ with a line $L_{i'}$ for $i' > i$. 

We then continue the path from $L$ until the intersection of $L_{i'}$ with $s_i^j$, and then extend it by the segment $s_i^j$. Let $p = L_{i'} \cap s_i^j$. Formally, we add the lines $L_1, \ldots , L_{i'-1}$, the line $L[p_{i'-1}, p] \subseteq L_{i'}$ and the segment $s_i^j$ into $\Le$.
\end{itemize}

We do this procedure iteratively for all endpoints $p$ of a line $L \in \L_{0}$ which are not connected to some other line in $\L_{0} \cup \Le$, i.e., if $\{p\} \cap (\L_{0} \cup \Le \setminus \{L\})=\emptyset$, for the so far computed set $\Le$. 


\section{Proofs from Section \ref{sec:PTAS-large-rectangles}}

\begin{proof}[Proof of Lemma~\ref{l-and-le-properties}]
We will start by showing that the set of lines $\L_{0} \cup \Le$ is nicely connected. From the construction of the lines it is clear that no two lines from $\L_{0} \cup \Le$ overlap or intersect properly. We need to show that for any line $L \in \L_{0} \cup \Le$ and any endpoint $p$ of $L$ there is a line $L' \in \L_{0} \cup \Le$ perpendicular to $L$ such that $L \cap L' = \{p\}$.

For each endpoint $p$ of a line $L \in \L_{0}$ which does not hit a perpendicular line from $\L_{0}\cup\Le$ we added a perpendicular line touching $p$ to the set $\Le$. The path of lines connecting $p$ with a line from $\L_{0}\cup\Le$ is constructed in such a way, that each line added to $\Le$ has both endpoints touching perpendicular lines from $\L_{0}\cup\Le$. If a line from $\Le$ gets extended, it is extended in such a way that the new endpoint touches a perpendicular line from $\L_{0}\cup\Le$. The set of lines $\L_{0}\cup\Le$ is nicely connected.

We will now show an upper bound on $|\Le|$. From Proposition \ref{prop:l-size} the set $\L_{0}$ consists of \mbox{$16(1 /\delta)^{2}+4$} lines. For each endpoint of a line from $\L_{0}$, except from the four lines in $\L_{0}$ bounding the input square, 
we added at most $M+1 = O(1/(\eps \delta^{2}))$ lines to the set $\Le$, which gives $|\Le| = O(1/(\eps \delta^4))$.

We will now upper bound the total weight of rectangles from $\R$ which are cut by a line from $\L_{0}\cup\Le$ parallel to their shorter edge. The only lines from $\L_{0}\cup\Le$ which cut rectangles along their shorter sides, i.e., which cut all blocks corresponding to a given rectangle, are the lines from $\Le$ which correspond to segments $s_{i}^{j}$. Each segment $s_{i}^{j}$ cuts only the set of rectangles $\R_{i}^{j}$, and we add to $\Le$ only such segments $s_i^j$, for which $w(\R_{i}^{j})\le\frac{2}{M}w(\R)$. The number of segments
$s_{i}^{j}$ added to $\Le$ is upper bounded by $32(1 / \delta)^2$. That gives an upper bound of $32(1/\delta)^2 \cdot \frac{2}{M}w(\R) \le \eps w(\R)$ (as $M = 64\frac{1}{\eps}(\frac{1}{\delta})^{2}$) on the total weight of rectangles from $\R$ which are cut by a line from $\L_{0}\cup\Le$ parallel to their shorter edge. 

From the construction of the lines $\Le$ it is clear that all lines from $\Le$ cutting rectangles in $\R$ lie on grid lines, and a line from $\Le$ cannot be contained in the interior of a grid cell.
\end{proof}

\begin{proof}[Proof of Lemma~\ref{lem:corridor-shapes}]

First we prove some additional lemmas.

\begin{lem}\label{lem:bottom-sticking-in}
Let $L \in \L_{0} \cup \Le$ be a line and let $Q$ be a grid cell such that $L$ is a sticking-in line (not necessarily extremal) for an edge $e$ of $Q$, and $L$ does not cross $Q$. Let $p \in Q$ be an endpoint of $L$. Then there exists a line $\bar{L} \in \L_{0}$ perpendicular to $L$ such that $L \cap \bar{L} = \{p\}$, $\bar{L}$ does not end at $p$ and $\bar{L}$ is an extremal sticking-in line for an edge $e'$ of $Q$ perpendicular to $e$.
\end{lem}

\begin{proof}
From Lemma \ref{l-and-le-properties} the set of lines $\L_{0} \cup \Le$ is nicely connected, which means that there is a line $\bar{L} \in \L_{0} \cup \Le$ perpendicular to $L$ such that $L \cap \bar{L} = \{p\}$. We will show that $\bar{L}$ satisfies the remaining conditions of the lemma statement. Without loss of generality we assume that $e$ is the bottom edge of $Q$.

Let $L_{max}$ be a maximal line containing $\bar{L}$ which does not intersect any blocks or lines from $\L_{0}$. We will show that $L_{max}$ is the bottom-most sticking-in line for an edge $e'$ of $Q$. As $L$ is a sticking-in line for $e$ and $L$ does not cross $Q$ (i.e. $L$ does not touch the edge of $Q$ parallel to $e$), $L$ cannot be extended at $p$. Either $L$ hits $\bar{L}$ at $p$, or $L$ hits a perpendicular block at $p$. In either case $L_{max}$ does not end at $p$.

We first show that $L_{max}$ is a long line (i.e. longer than the grid granularity). If $\bar{L} \in \L_{0}$, then $\bar{L}$ is long and so is $L_{max}$. If $\bar{L} \in \Le$ then, from the construction of $\Le$, $\bar{L}$ goes along a long edge of a block, and $L_{max}$ is long as it contains a long edge of a block. 

\begin{figure}
\begin{centering}
\includegraphics[width=.3\textwidth]{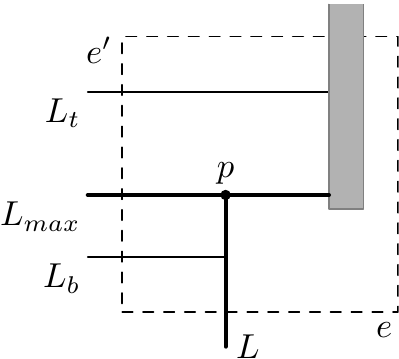}
\par\end{centering}

\caption{\label{fig:sticking-in-line}In the proof of Lemma \ref{lem:bottom-sticking-in} the line $L_{max}$ is the bottom-most sticking-in line for $e'$, and so it belongs to the set $\L_{0}$.}
\end{figure}

$L_{max}$ is not contained in $Q$, i.e., it intersects an edge $e'$ of $Q$ perpendicular to $e$ (see Figure \ref{fig:sticking-in-line}). Let $L_b \in \L_{0} \cup \Le$ be a line which intersects $e'$ below $L_{max} \cap e'$. As $L_b$ cannot intersect $L$ and $L_{max}$ extends beyond $L$, we get $|L_b \cap Q| < |L_{max} \cap Q|$.

If $L_{max}$ crosses $Q$, then it is the bottom-most line intersecting $e'$ and maximizing the length of the intersection with $Q$. Assume that $L_{max}$ does not cross $Q$. As $L_{max}$ is a maximal line which does not intersect any blocks or lines from $\L_{0}$, it ends in $Q$ by hitting a perpendicular line from $\L_{0}$  or a perpendicular block. This line or block does not intersect the bottom boundary of $Q$, as it would yield a long line crossing $e$ which reaches further than $L$, which gives a contradiction, as $L$ is a sticking-in line for $e$. The line or block hit by $L_{max}$ crosses the top edge of $Q$ and does not intersect any lines from $\L_{0} \cup \Le$. Therefore any line $L_t \in \L_{0} \cup \Le$ which intersects $e'$ above $L_{max} \cap e'$ satisfies $|L_t \cap Q| \le |L_{max} \cap Q|$.


We get that $L_{max}$ is the bottom-most long line maximizing the length of the intersection with $Q$, and so it is the bottom-most sticking-in line for $e'$. We get that $L_{max} \in \L_{0}$, and so $\bar{L} = L_{max}$. $\bar{L}$ satisfies all conditions of the lemma statement.
\end{proof}

\begin{lem}\label{lem:faces-line-continuation}
Let $Q$ be a grid cell, and let $e$ be an edge of $Q$. Let $L,L'\in\L_{0}\cup\Le$ be two lines intersecting $Q$, touching $e$ at $p_L$ and $p_{L'}$ respectively, such that there is no line $L'' \in \L_{0}\cup\Le$ which intersects $Q$ and touches $e$ between $p_L$ and $p_{L'}$. Then there is an edge $e' \neq e$ of $Q$ and lines $\bar{L},\bar{L}'\in\L_{0}\cup\Le$ touching $e'$ such that $L\cap\bar{L}\ne\emptyset$,
$L'\cap\bar{L}'\ne\emptyset$. 
\end{lem}

\begin{proof} First observe that the lemma statement allows that $L=\bar{L}$ or $L'=\bar{L}'$.

Assume w.l.o.g.~that $e$ is the bottom edge of $Q$, $p_L$ is on the left of $p_{L'}$, and that $|L\cap Q| \ge |L'\cap Q|$. 
We have to consider three cases. We start with the most interesting case, where
both $L$ and $L'$ do not cross $Q$. Let $p$ and $p'$ be the endpoints of $L$ and $L'$, respectively, in $Q$. Let $\tilde{L}$ and $\tilde{L}'$ be two lines from $\L_{0} \cup \Le$ perpendicular to $L$ and $L'$ such that $L\cap\tilde{L}=\{p\}$ and $L'\cap\tilde{L}'=\{p'\}$, respectively. From Lemma \ref{l-and-le-properties} such lines exist and they are not contained in $\int(Q)$, i.e., each of them touches an edge of $Q$. If $\tilde{L}$ touches the right edge of $Q$, then $\tilde{L}'$ also touches the right edge of $Q$ and we are done  (see Figure \ref{fig:corridor-in-cell}a). If $\tilde{L}'$ touches the left edge of $Q$, then $|L\cap Q| = |L'\cap Q|$, $\tilde{L} = \tilde{L}'$ and we are done  (see Figure \ref{fig:corridor-in-cell}b).
 
\begin{figure}
\begin{centering}

\centerline{%
\includegraphics[height=0.25\textwidth]{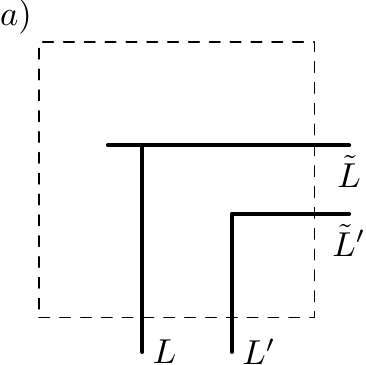}%
\hskip0.05\textwidth
\includegraphics[height=0.25\textwidth]{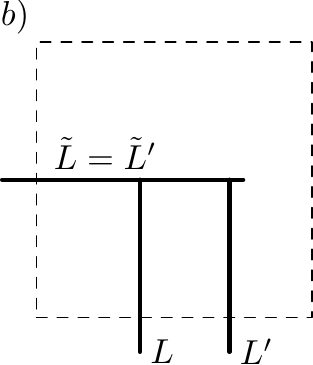}%
\hskip0.05\textwidth
\includegraphics[height=0.25\textwidth]{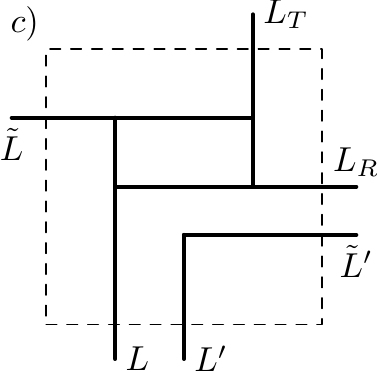}%
}%
\vskip0.03\textwidth
\centerline{%
\includegraphics[height=0.25\textwidth]{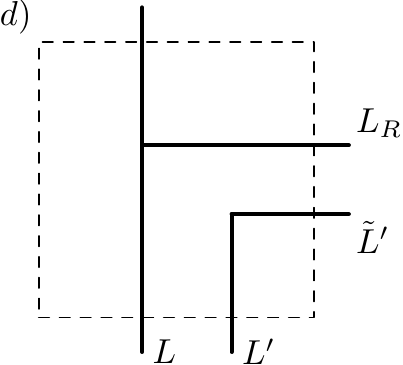}%
\hskip0.05\textwidth
\includegraphics[height=0.25\textwidth]{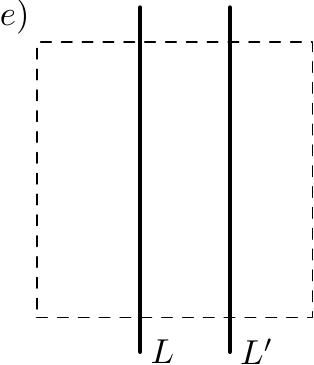}%
}%
\par\end{centering}
\caption{\label{fig:corridor-in-cell} Neighboring lines $L,L' \in \L_{0} \cup \Le$ touch perpendicular lines which intersect the same edge of $Q$.}
\end{figure}

The only remaining possibility is that $\tilde{L}$ touches only the left edge of $Q$, and $\tilde{L'}$ only the right edge of $Q$  (see Figure \ref{fig:corridor-in-cell}c). As there are no edges in $\L_{0} \cup \Le$ intersecting $Q$ and touching $e$ in between $p_L$ and $p_{L'}$, $L$ is a sticking-in line for the edge $e$ of $Q$. From Lemma \ref{lem:bottom-sticking-in} $\tilde{L}$ is the bottom-most sticking-in line for the left edge of $Q$, and it does not end at $p$. $\tilde{L}$ does not touch the right edge of $Q$, and applying Lemma \ref{lem:bottom-sticking-in} to $\tilde{L}$ gives us that $\tilde{L}$ hits a perpendicular sticking-in line $L_T$ in $Q$. $L_T$ does not touch the bottom edge $e$ of $Q$, as we would have $|L_T \cap Q| > |L \cap Q|$, and $L$ is a sticking-in line. $L_T$ touches the top edge of $Q$, and $L_T$ is a sticking-in line for this edge. As $L_T$ ends in $Q$, applying Lemma \ref{lem:bottom-sticking-in} to $L_T$ gives, that it hits a perpendicular (i.e. touching the right edge of $Q$) 
sticking-in 
line $L_R$. Either $L_R = \tilde{L}'$, or $L_R$ is above $\tilde{L}'$, so $L_R$ does not hit $L'$, or any line to the right of $L'$. Applying Lemma \ref{lem:bottom-sticking-in} to $L_R$ gives that $L_R$ hits a perpendicular sticking-in line, and the only candidate for such a line hit by $L_R$ is $L$. $L$ touches $L_R$ and $L_R$ touches the right edge of $Q$, and we are done.

In the second case the line $L$ touches the top edge of $Q$, and $L'$ does not. Let $\tilde{L}' \in \L_{0} \cup \Le$ be the perpendicular line touching $L'$ at its endpoint $p' \in Q$. From Lemma \ref{l-and-le-properties} such line exists and is not contained in $\int(Q)$. As $L$ is to the left of $L'$, $\tilde{L}'$ touches the right edge of $Q$ (see Figure \ref{fig:corridor-in-cell}d). Let $L_R$ be a sticking-in line for the right edge of $Q$. Such line exists, as the maximal line containing $\tilde{L}'$ is a candidate for it. We will show that $L_R$ touches $L$. If $L_R$ crosses $Q$, then $L_R$ must touch $L$ (and either $L$ or $L_R$ goes along an edge of $Q$). If $L_R$ does not cross $Q$, from Lemma \ref{lem:bottom-sticking-in} we get that $L_R$ hits a perpendicular sticking-in line in $Q$. As $L$ is the rightmost line crossing $Q$ (all lines to the right of $L$ cannot exceed $\tilde{L}'$), it is the rightmost sticking-in line and $L_R$ touches $L$. We set $\bar{L}=L_R$.

In the last case, when both $L$ and $L'$ touch the upper edge of $Q$, the claim is immediate (see Figure \ref{fig:corridor-in-cell}e).
\end{proof} 

First we will consider the case when $C$ has non-empty intersection with some block $B \in \B$ contained in $F$. Let $e$ be an edge of $Q$ such that $e \cap \int(B) \neq \emptyset$. Assume w.l.o.g. that $e$ is the bottom edge of $Q$ (see Figure \ref{fig:corridor-in-cell-2}a). Let $L, L' \in \L_{0}\cup\Le$ be lines which intersect $Q$ and touch $e$ at some points $p_L$ and $p_{L'}$, respectively, such that $p_L$ is to the left of $e \cap \int(B)$, $p_{L'}$ is to the right of $e \cap \int(B)$, and no line from $\L_{0}\cup\Le$ which intersects $Q$ touches $e$ in between $p_L$ and $p_{L'}$. Such lines exist, as no line from $\L_{0}\cup\Le$ intersects $e$ inside $e \cap \int(B)$, the leftmost long line intersecting $Q$ and touching $e$ (which belongs to $\L_{0}$) either contains the left edge of $B$ or is to the left of it, and the rightmost long line intersecting $Q$ and touching $e$ (which also belongs to $\L_{0}$) either contains the right edge of $B$ or is to the right of it.

\begin{figure}
\begin{centering}

\centerline{%
\includegraphics[height=0.3\textwidth]{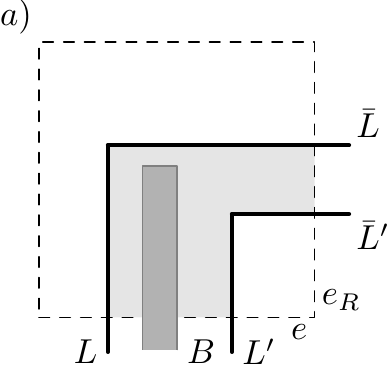}%
\hskip0.05\textwidth
\includegraphics[height=0.3\textwidth]{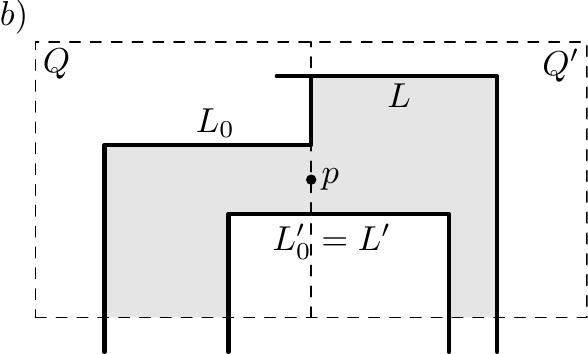}%
}%

\par\end{centering}
\caption{\label{fig:corridor-in-cell-2} A connected component of a face $F \in \F_+(\L_{0} \cup \Le)$ within a grid cell (denoted by a shaded area) must have a simple shape, i.e., it is either a rectangle or an L-shape.}
\end{figure}
  
Parts of the lines $L, L'$ lie on the boundary of $C$. Denote by $\bar{L}$ and $\bar{L}'$ the lines given by applying Lemma~\ref{lem:faces-line-continuation} to $L$ and $L'$. If $\bar{L}$ and $\bar{L}'$ both intersect the top edge of $Q$ then the claim follows and in particular $\int(C)$ is the interior of a rectangle. Otherwise, assume w.l.o.g.~that they intersect the right edge $e_{R}$
of $Q$ and assume w.l.o.g.~that $\bar{L}$ is the bottommost line
touching $L$ and $e_{R}$ and $\bar{L}'$ is the topmost line
touching $L'$ and $e_{R}$. From Lemma \ref{l-and-le-properties} 
the set of lines $\L_{0} \cup \Le$ is nicely connected, and by construction, all lines in $\L_0 \cup \Le$ with non-empty intersection with
$\int(Q)$ for some grid cell $Q$ touch the boundary of Q.
Hence, there can be no line in $\L_{0}\cup\Le$ intersecting
$e_{R}$ between $e_{R}\cap\bar{L}$ and $e_{R}\cap\bar{L}'$. Hence, the
claim follows.

We already know that the lemma holds for any connected component $C$ of $F \cap Q$ for any grid cell $Q$ such that $C$ has non-empty intersection with $\int(B)$ for some block $B \in \B$ contained in $F$. Now we will show that if the lemma is satisfied for some connected component $C$ of $F \cap Q$, then it is also satisfied for a connected component $C'$ of $F \cap Q'$ if $C \cap C' \neq \emptyset$ (see Figure \ref{fig:corridor-in-cell-2}b). That will prove the lemma.

Let $e=Q \cap Q'$, and let $p \in C \cap C'$. Let $L_0, L'_0 \in \L_{0} \cup \Le$ be lines bounding $C$ and touching $e$. From the construction above we know that such lines exist. We want to show that there are two lines $L, L' \in \L_{0} \cup \Le$ intersecting $Q'$ and touching $e$ such that $p \in L(e \cap L, e \cap L')$ and there are no lines intersecting $Q'$ and touching $e$ in $L(e \cap L, e \cap L')$. Then, by proceeding exactly as in the first case, we prove the lemma.

Assume w.l.o.g. that $e$ is a vertical edge and that $L_0$ is above $L'_0$. We show only that there is a line $\bar{L}$ intersecting $Q'$ and touching $e$ above $p$ (with a similar reasoning one can show that there is a line intersecting $Q'$ and touching $e$ below $p$). Let $\bar{L}_0$ be the top-most horizontal line in $\L_0$ intersecting $Q$ and touching $e$. We will first show that such line exists and it is not below $L_0$. If $L_0 \in \L_{0}$, then either $\bar{L}_0 = L_0$ or $\bar{L}_0$ is above $L_0$ and we are done. Otherwise $L_0 \in \Le$ and, from the construction of $\Le$, $L_0$ goes along a long edge of a block, and the maximal line $L_{max}$ containing $L_0$ which does not intersect blocks and lines from $\L_0$ is long. As $L_{max} \notin \L_{0}$, there must be a line $\bar{L}_0 \in \L_{0}$ which is above $L_{max}$.

If $\bar{L}_0$ intersects $Q'$, we set $\bar{L}:=\bar{L}_0$ and we are done. Otherwise, observe that $\bar{L}_0 \in \L_0$ and hence it is maximal. We claim that $\bar{L}_0$ hits a perpendicular line $L_T \in \L_0$ at $\bar{L}_0 \cap e$ (and, as the name suggests, it will turn out that $L_T$ crosses the top edge of $Q'$). 
As $\bar{L}_0$ is maximal, it has to hit a perpendicular block $B$ or a perpendicular line from $\L_{0} \cup \Le$ at its endpoint in $Q'$. If the first case occurs, then the maximal long line going along the left edge of $B$ is the leftmost long line intersecting the top edge of $Q'$, and it belongs to $\L_{0}$. We get that $\bar{L}_0$ hits a perpendicular line $L_T \in \L_{0} \cup \Le$ at its endpoint in $Q'$. We cannot have $L_T \in \Le$, as the maximal long line containing $L_T$ is the leftmost long line intersecting the top edge of $Q'$, and so it belongs to $\L_{0}$. We have that  $L_T \in \L_{0}$.

As $p\in C\cap C'$ we have that $p \notin L_T$. Thus, $L_T$ ends
above $p$ by hitting 
a perpendicular line $\bar{L} \in \L_{0} \cup \Le$ or
a perpendicular block $B'$.
In the first case we are done. In the second case,
the topmost long line $\tilde{L}$ intersecting $Q'$ and touching $e$ contains the upper edge of $B'$ and crosses $e$ above $p$. Hence, we can set $\bar{L}:=\tilde{L}$ and we are done.
\end{proof}

\begin{proof}[Proof of Lemma~\ref{lem:corridor-at-boundaries}]
Let $p \in C\cap e$. Let $C'$ be a connected component of $F \cap Q'$ containing $p$. We get that $C \cap C' \neq \emptyset$, and so there exists a component $C'$ which satisfies the desired properties.

\begin{figure}
\begin{centering}
\includegraphics[height=.3\textwidth]{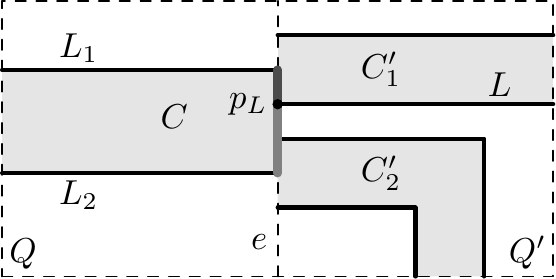}
\par\end{centering}

\caption{\label{fig:no-forks} Lemma~\ref{lem:corridor-at-boundaries} shows that it cannot happen that a connected component $C$ of $F \cap Q$ "forks" into two connected components $C_1'$ and $C_2'$ in a neighboring grid cell $Q'$, as shown in this figure. In particular, one of the gray segments on $e$ must be contained in $\L_{0} \cup \Le$.}
\end{figure}

We now show that $C'$ is unique. Assume otherwise, i.e. that there are two connected components $C_1'$ and $C_2'$ of $F \cap Q'$ which have non-empty intersection with $C$. Let $L \in \L_{0} \cup \Le$ be a line intersecting $Q'$, such that $L \cap e$ is between $C_1' \cap C$ and $C_2' \cap C$ (see Figure \ref{fig:no-forks}). Such line exists, as $C_1' \cap Q'$ and $C_2' \cap Q'$ are not connected. Let $L_1, L_2 \in \L_{0} \cup \Le$ be the lines intersecting $Q$, touching $e$ and bounding $C \cap Q$. Then there is no line $L_3 \in \L_{0} \cup \Le$ which intersects $Q$ and touches $e$ between $L_1 \cap e$ and $L_2 \cap e$. As $L \cap e$ is between $C_1' \cap C$ and $C_2' \cap C$, it holds that $L \cap e$ is between $L_1 \cap e$ and $L_2 \cap e$. That gives us that $L$ cannot intersect $Q$, and so $L$ has an endpoint $p_L \in e$. We will now show that 
one of the segments $L[p_L,L_1 \cap e]$ and $L[p_L,L_2 \cap e]$ is contained in $\L_{0} \cup \Le$,
which gives contradiction, as it requires $C_1' \cap C = \emptyset$ 
or $C_2' \cap C = \emptyset$. Hence, the component $C'$ is unique.

We will now show the following lemma.

\begin{lem}\label{lem:l-on-grid-boundary}
Let $Q$ and $Q'$ be two neighboring grid cells and let $e = Q \cap Q'$. Let $s \subseteq e$ be a maximal segment of $\L_{0} \cup \Le$ contained in $e$, and assume that $s$ does not contain any endpoint of $e$. Then $s$ is incident with lines $L, L' \in \L_{0} \cup \Le$ (where possibly $L = L'$) such that $L$ intersects $Q$ and $L'$ intersects $Q'$.
\end{lem}
\begin{proof}
From the construction of $\L_{0}$ and $\Le$, the segment $s$ consists of one or multiple segments $s_i^j$, as any other line from $\L_{0} \cup \Le$ would touch an endpoint of $e$. Let $s_{i_0}^{j_0} \subseteq s$ be the first one added to $\Le$. W.l.o.g. assume that in the construction of the path, the line $L \in \Le$ preceding $s_{i_0}^{j_0}$ on the path intersects $Q$. Let $L' \in \L_{0} \cup \Le$ be the line perpendicular to $e$ which is at the other end of $s_{i_0}^{j_0}$. If $L'$ intersects $Q'$, we are done. Assume, for contradiction, that $L'$ intersects $Q$ and has an endpoint at $p' \in e$. We have that $L'$ does not hit a perpendicular line from $\L_{0} \cup \Le$ at $p'$ as otherwise such a line would be contained in the leftmost long line crossing the top or the bottom edge of $Q'$ which are by definition in $\L_0$ and the segment $s_{i_0}^{j_0}$ has by construction non-empty intersection with all lines in $\L_{0} \cup \Le$.
Assume that $L' \in \Le$. From the construction of $\Le$, in particular from Lemma \ref{lem:line-hits-block}, we get that the successor of $L'$ on the path must be a segment $s_i^j$ (as if $L'$ was not "shortened", it would end in the interior of a grid cell). Then $s_{i_0}^{j_0}$ is neighboring to the segment $s_{i}^{j}$, which has been added to $\Le$ before $s_{i_0}^{j_0}$. As then $s_{i}^{j} \subseteq s$ from the maximality of $s$, we get a contradiction, as  $s_{i_0}^{j_0}$ was the first segment from $s$ added to $\Le$.

Assume that $L' \in \L_{0}$. Then, as $L'$ is maximal and it does not hit a perpendicular line from $\L_{0} \cup \Le$ at $p'$, $L'$ hits a perpendicular block $B \in \B$ at $p'$. The long edge of $B$ containing $p'$ is then an extremal long line crossing an edge of $Q$ perpendicular to $e$, and do it belongs to $\L_{0}$. We get a contradiction, as $L'$ does not hit a perpendicular line from $\L_{0}$ at $p'$.

We obtain that $L'$ intersects $Q'$, which proves the lemma.
\end{proof}

As $p_L$ is an endpoint of $L$, $L$ touches a perpendicular line from $\L_{0} \cup \Le$ at $p_L$. 
Let $s$ be a maximal segment of $e$ containing $p_L$ and contained in $\L_{0} \cup \Le$. From Lemma \ref{lem:l-on-grid-boundary} segment $s$ contains an endpoint of $e$, or $s$ is incident with a line intersecting $Q$. In any case, one of the segments $L[p_L,L_1 \cap e]$, $L[p_L,L_2 \cap e]$ is contained in $s$, i.e. it is contained in $\L_{0} \cup \Le$.
\end{proof}

\begin{proof}[Proof of Lemma~\ref{lem:maze-properties}]
From Proposition \ref{prop:l-size} we get that $|\L_{0}| = O((\frac{1}{\delta})^{2})$. From Lemma \ref{l-and-le-properties} $|\Le|\le\frac{1}{\eps}\cdot(\frac{1}{\delta})^{O(1)}$.
For any line $L \in \L_{0} \cup \Le$ the number of rectangles intersected, but not cut by $L$ is at most $2$. As the length of $L$ is at most $N$, the number of rectangles $R$ cut by $L$ such that $|L \cap R| > \delta N$ is smaller than $1/\delta$. The number of circumvented polygons is therefore at most $\frac{1}{\eps}\cdot(\frac{1}{\delta})^{O(1)}$.
Circumventing a polygon generates four lines in $\L$, and possibly splits some lines from $\L_{0}\cup\Le$ in two. The number of lines in $\L$ is at most $\frac{1}{\eps}\cdot(\frac{1}{\delta})^{O(1)}$.

The added lines do not intersect rectangles. The only rectangles intersected by a line in $\L$ are the rectangles cut by a line from $\L_{0}\cup\Le$ parallel to their shorter edge. From Lemma \ref{l-and-le-properties} the total weight of such rectangles is upper bounded by $\eps\cdot w(\R)$.
\end{proof}

\begin{figure}
\centerline{\includegraphics{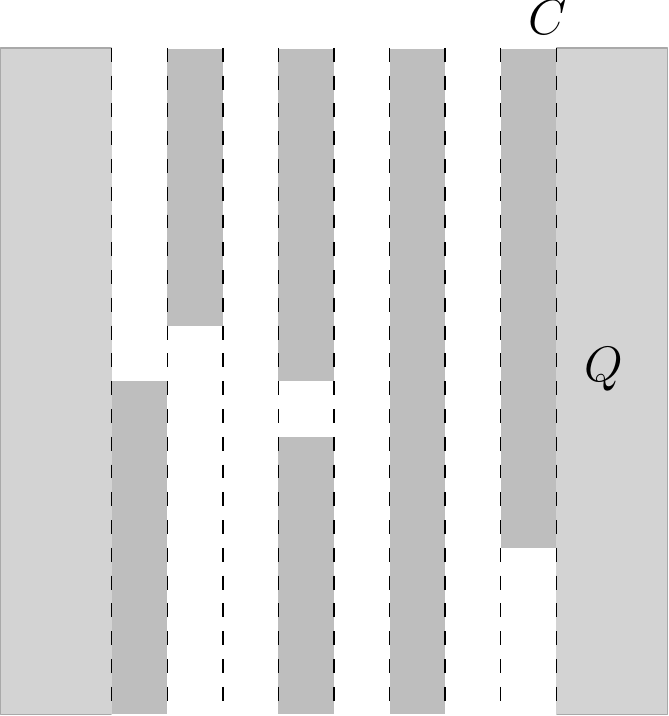}\hskip0.05\textwidth\includegraphics{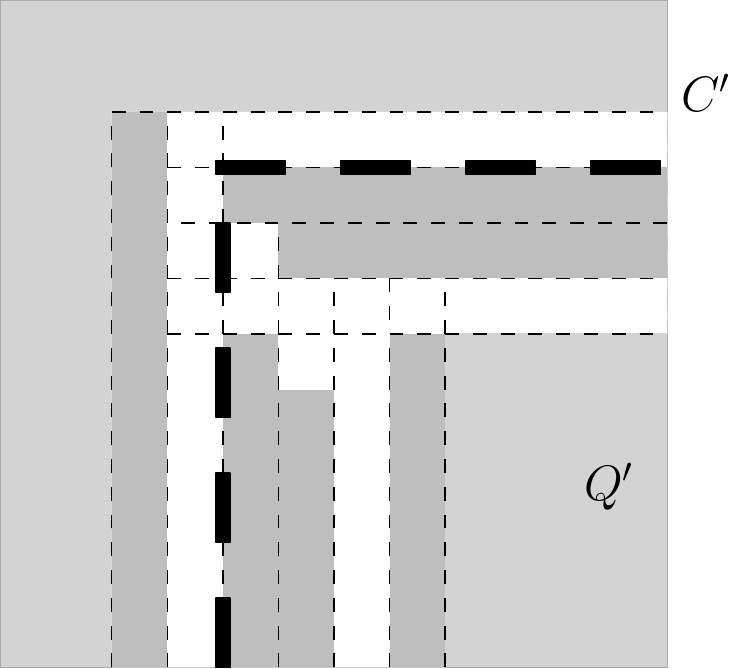}
}

\caption{\label{fig:slice-corridors}Left: the dashed lines denote the lines in the set $\L(C)$ for the
component $C$ in the cell~$Q$. Right: the dashed lines represent
a ground set from which the pairs in the set $\L(C)$ are created.
The bold pair of lines represents an example of an element in $\L(C)$.}
\end{figure}

\begin{proof}[Proof of Lemma~\ref{lem:dp-for-paths}]
We first prove the lemma for the setting that each rectangle is a
block, i.e., for each $R_i$ we have $h_{i}=1$ or $g_{i}=1$. Consider
a face $F \in \F_+(\L)$ which is homeomorphic to a line. For each grid cell
$Q$ and for each connected component $C$ of $F\cap Q$ we will define a set of lines $\L(C)$. Due to Lemma~\ref{lem:corridor-shapes}
the component $C$ is either a rectangle or an L-shape (note that here we can have a "degenerated L-shape" with only four edges). If $C$ is
a rectangle, we consider any edge $e$ of $Q$ such that $C \cap e \neq \emptyset$, and we define $\L(C)$ to be the set of all maximal
lines perpendicular to $e$, with integer endpoints, which are contained in $C$ or in the boundary of $C$ (see
Figure~\ref{fig:slice-corridors}). Observe that if all rectangles are blocks,
the lines in $\L(C)$ do not intersect any rectangles. Now consider
the case when $C$ is an L-shape and let $e,e'$ be the (perpendicular) edges of $Q$
such that $C\cap e\ne\emptyset\ne C\cap e'$. (Notice that we can have a special case that $C$ intersects only one edge $e$ of $Q$, i.e., the face $F$ does not extend beyond $C$. However, then we take as $e'$ the edge of $Q$ perpendicular to $e$ which contains the boundary edge of $C$ perpendicular and non-adjacent to $C \cap e$. Intuitively, that is the edge where the face would continue beyond $C$ if it was not blocked by a line from $\L$.) We define $\L(C)$ to
be the set of all pairs of straight lines $(L,L')$ with integer endpoints contained
in $C$, which do not intersect any blocks, and such that $L$ is perpendicular to $e$ and $L'$ to $e'$, and there is a point $p$ such that $L$ has one endpoint at $e$ and the other one at $p$, and $L'$ has one endpoint at $e'$ and the other one at $p$ (see Figure~\ref{fig:slice-corridors}). 

Next, we define a family of faces $\F(F)$ contained in the face $F$. Each face $F'\in\F(F)$
will have bounded complexity (i.e., at most $\frac{1}{\eps}\cdot(\frac{1}{\delta})^{O(1)}$
boundary lines). We will construct $\F(F)$ in such a way that 
$F \in \F(F)$ and each
face in $\F(F)$ which is not a rectangle of unit height or width 
can be decomposed into a bounded number of subfaces
in $\F(F)$ without intersecting any block. During the recursion of
GEO-DP, when parametrized with large enough $k$, the algorithm will consider exactly this decomposition.
Thus, GEO-DP optimally solves the subproblem induced by the face $F$.

\begin{figure}
\centerline{\includegraphics{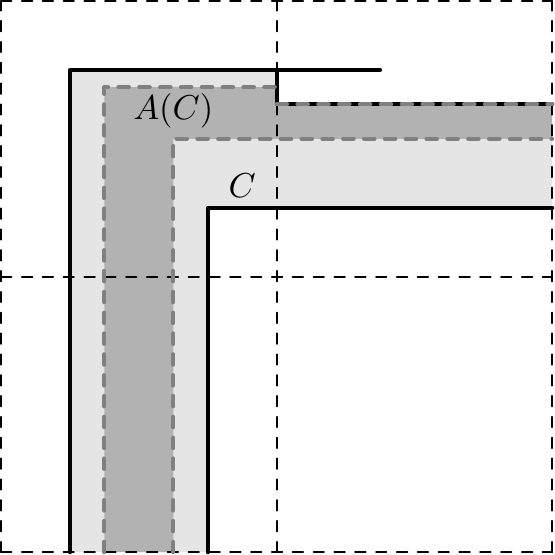}
}
\caption{\label{fig:subfaces}An element of $\F(F)$ (dark gray) for a face $F$ (light gray). The element of $\F(F)$ is a union of subareas $A(C)$, one for each connected component $C$ of $F$ within a grid cell. The chosen subareas $A(C)$ for different components $C$ are consistent.}
\end{figure}

Now we define the family $\F(F)$. For any grid cell $Q$ and for any connected component $C$ of $F \cap Q$
we select two elements from the set of lines $\L(C)$, allowing to select the same
element twice,
but not allowing the two chosen elements to intersect properly. Let $A(C)$ denote the subarea of $C$ which is strictly between
the two selected elements from $\L(C)$ (i.e., the chosen elements from $\L(C)$ do not belong to $A(C)$), see Figure~\ref{fig:subfaces}. 
We require
the selected elements for the different components to be \emph{consistent}, meaning that for any two connected components $C, C'$ of $F$ within some grid cells such that $C \cap C' \neq \emptyset$ we have 
$A(C) \cap C \cap C' = A(C') \cap C \cap C'$
(i.e., the subareas chosen for different grid cells match at the boundaries of the grid cells), and $\bigcup_{C}A(C)$ is connected.
For any choice of consistent elements for all
connected components $C$
we add $F':=\bigcup_{C}A(C)$ to $\F(F)$.
Due to the definition of $\F(F)$, every
face in $\F(F)$ which is not a rectangle of unit height or width can be decomposed in two disjoint elements of $\F(F)$. Every face in $\F(F)$ which is a rectangle of unit height or width contains at most $\frac{1}{\delta}$ blocks, and GEO-DP finds an optimal solution for it.

The decomposition can be seen as follows. Let $F'\in \F(F)$. Take a component $C$ of $F'$ within a grid cell $Q$ which is a dead-end, i.e, there is only one edge $e$ of $Q$ such that $C\cap e \ne \emptyset$. Draw a line $L_0$ splitting $C$ into two components without intersecting any block, such that one end of the line touches the boundary of $F'$ (if $C$ is an L-shape we might need two lines for that). We take the loose end of $L_0$ and extend it until we touch the boundary of $F'$ or we hit a perpendicular block. In case we touch the boundary of $F'$, we are done. If $L_0$ hits a block $B$ at some point $p$, we continue similarly as in case of loose ends of the lines in $\L_{0}$. 
We draw a new line $L_1$ starting at $p$ and following the edge of $B$ so that we cross a grid cell boundary. We continue iteratively until we draw a line $L_m$ whose end touches the boundary of $F'$. Then the lines $L_0,...,L_m$ define a path splitting $F'$ into two subpaths. 
See Figure~\ref{fig:slice-corridors-cycles}.

It remains to upper bound the complexity of the faces in $\F(F)$, i.e.,
bound the number of their boundary edges. By definition, the latter quantity
is in the order of the number of connected components of $F\cap Q$
for all grid cells $Q$. The reason is that for each face $F'\in\F(F)$
the boundary of a connected component of $F'\cap Q$ (for any grid cell $Q$) has only a constant number edges.
As the number of lines in $\L$ is upper bounded by $\frac{1}{\eps}\cdot(\frac{1}{\delta})^{O(1)}$ (see Lemma~\ref{lem:maze-properties})
and there are $O((\frac{1}{\delta})^{2})$ grid cells, the number
of such components $C$ is upper bounded by $\frac{1}{\eps}\cdot(\frac{1}{\delta})^{O(1)}$. 

We conclude that if all rectangles are blocks, GEO-DP finds an optimal
solution for $F$ if $k\ge\frac{1}{\eps}\cdot(\frac{1}{\delta})^{O(1)}$.
For the case of arbitrary rectangles we observe that the boundary
of each face in $\F(F)$ intersects rectangles from $\R$ contained in $F$ only parallel
to their longer edges. Hence, for any face $F'\in\F(F)$ there are only
$\frac{1}{\eps}\cdot(\frac{1}{\delta})^{O(1)}$ rectangles $\R(F')$
which are contained in $F$ and intersected by the boundary edges of $F'$. Hence, whenever in the above argumentation
we decompose a face $F'\in\F(F)$ into two disjoint elements $F_{1},F_{2}\in\F(F)$,
we can instead argue that the face $F'\setminus\R(F')$ is decomposed
into the faces $F_{1}\setminus\R(F_{1})$ and $F_{2}\setminus\R(F_{2})$
and the at most $\frac{1}{\eps}\cdot(\frac{1}{\delta})^{O(1)}$ rectangles
$(\R(F_{1})\cup\R(F_{2}))\setminus\R(F')$. Also, again each face $F\setminus\R(F)$
has a boundary with at most $\frac{1}{\eps}\cdot(\frac{1}{\delta})^{O(1)}$
edges. 
Note that in this case $F\setminus\R(F)$ can consist of multiple (at most $\frac{1}{\eps}\cdot(\frac{1}{\delta})^{O(1)}$) 
connected components and GEO-DP recurses on each of them separately.
Hence, in the case of arbitrary rectangles, GEO-DP parametrized with $k\ge\frac{1}{\eps}\cdot(\frac{1}{\delta})^{O(1)}$ finds
an optimal solution for $F$.
\end{proof}

\begin{proof}[Proof of Lemma~\ref{lem:dp-for-cycles}]

Like in the proof of Lemma~\ref{lem:dp-for-paths} above, let us
first assume that all rectangles are blocks. Consider a face
$F\in\F_{+}(\L)$ which is homeomorphic to a cycle. 
We now describe a procedure to draw a path consisting of a set of lines
$L_{0},L_{1},...,L_{\ell}$, which will either add further structure to $F$ so that it
becomes a path-face
or will subdivide $F$ into a cycle-face and a path-face.
The endpoints of the lines will have integer coordinates, and their number $\ell$ can be arbitrarily large (in particular, larger than $\frac{1}{\eps}\cdot(\frac{1}{\delta})^{O(1)}$ etc.). 
The lines will not intersect any blocks contained in $F$.

\begin{figure}
\begin{centering}
\includegraphics[height=0.25\textwidth]{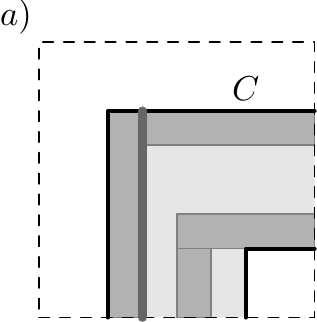} \hskip0.05\textwidth
\includegraphics[height=0.25\textwidth]{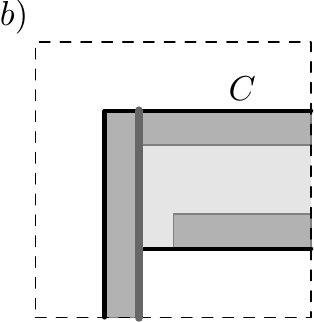} 
\par\end{centering}

\caption{\label{fig:slice-cycle-step-one}Subdividing an L-shape into a
rectangle and an L-shape (case $a)$) or into two rectangles (case
$b)$). The splitting line $L$ is depicted in dark gray. Notice that $L$ does not intersect any blocks (pictured as gray rectangles).}
\end{figure}

Consider
an arbitrary connected component $C$ of $F \cap Q$ for some grid cell $Q$, such that $C$ is an L-shape. Such a component
exists since otherwise $F$ would not be a cycle. It is always possible
to draw a line $L$ within $C$ which does not intersect any blocks
and which subdivides $C$ into a rectangle and an L-shape or into two rectangles (see Figure \ref{fig:slice-cycle-step-one}). In the second case the line $L$ transforms $F$ into a path-face and we are done. Now consider the first case. We define $L_{0}$ as a maximal line containing $L$
which does not
intersect any block and which does not cross the boundary of $F$ or overlap an edge from the boundary of $F$. If both endpoints of $L_{0}$ touch the boundary
of $F$, we are done. Otherwise, we continue very similarly as
in the construction of the lines $\Le$. Let $p_{0}$ be an endpoint
of $L_{0}$ which hits a perpendicular block $B \in \B$. Observe that this
must happen within a connected component $C'$ of $F$ where $C'$
is an L-shape. As $B$ is large, one of its endpoints must lie outside of
$C'$. Denote by $L_{\max}$ the maximal line within $F$ containing
$p_{0}$ and going along the edge of $B$ which does not intersect any blocks. Denote by $p_{1}$
the endpoint of $L_{\max}$ which is outside of $C'$. We define $L_{1}:=L[p_{0},p_{1}]$
and continue iteratively. Observe that in contrast to the definition
of the lines $\Le$, we do not stop after some fixed number of iterations,
but continue until the added line $L_{\ell}$ either touches the boundary
of $F$, or hits some line $L_{i}$ with $i<\ell$. Let $\L_{F}:=\{L_{0},...,L_{\ell}\}$.
If the path touches both boundaries of the cycle-face $F$, it transforms $F$ into a path-face. Otherwise (i.e., if the path touches the same boundary twice or ends by hitting itself), the path subdivides $F$ into a path-face $F'$ and a cycle-face $F''$.

As in the proof
of Lemma~\ref{lem:dp-for-paths}, we can formulate $F''$ as the union
of sets $A(C)$, where for each component $C$ we select two elements
from the set $\L(C)$. For each component $C$ we do this as follows.
If $C$ is not intersected by the path $\L_{F}$, we choose the two elements of $\L(C)$ going along the boundaries of $C$. Otherwise, if 
$L_\ell$ does not end in $C$, we select the element in $\L(C)$ which is given by the last line(s) in $\L_F$ crossing $C$ 
and one element $\L(C)$ which describes the boundary of $C$ not touched by $\L_{F}$.
If $L_\ell$ ends in $C$, we also select the boundary of $C$ not touched by $\L_{F}$, and additionally an element in $\L(C)$ which consists of $L_\ell \cap C$ and a segment of of $L_i \cap C$ where $L_i$ denotes the penultimate line in $\L_{F}$ crossing $C$.
With a similar reasoning as in the proof of
Lemma~\ref{lem:dp-for-paths}, we can then upper bound the complexity of
the boundaries of $F''$ by $\frac{1}{\eps}\cdot(\frac{1}{\delta})^{O(1)}$.
See Figure~\ref{fig:slice-cycle-faces-many-paths} for an example of the described operation.

\begin{figure}
\centerline{\includegraphics[scale=0.7]{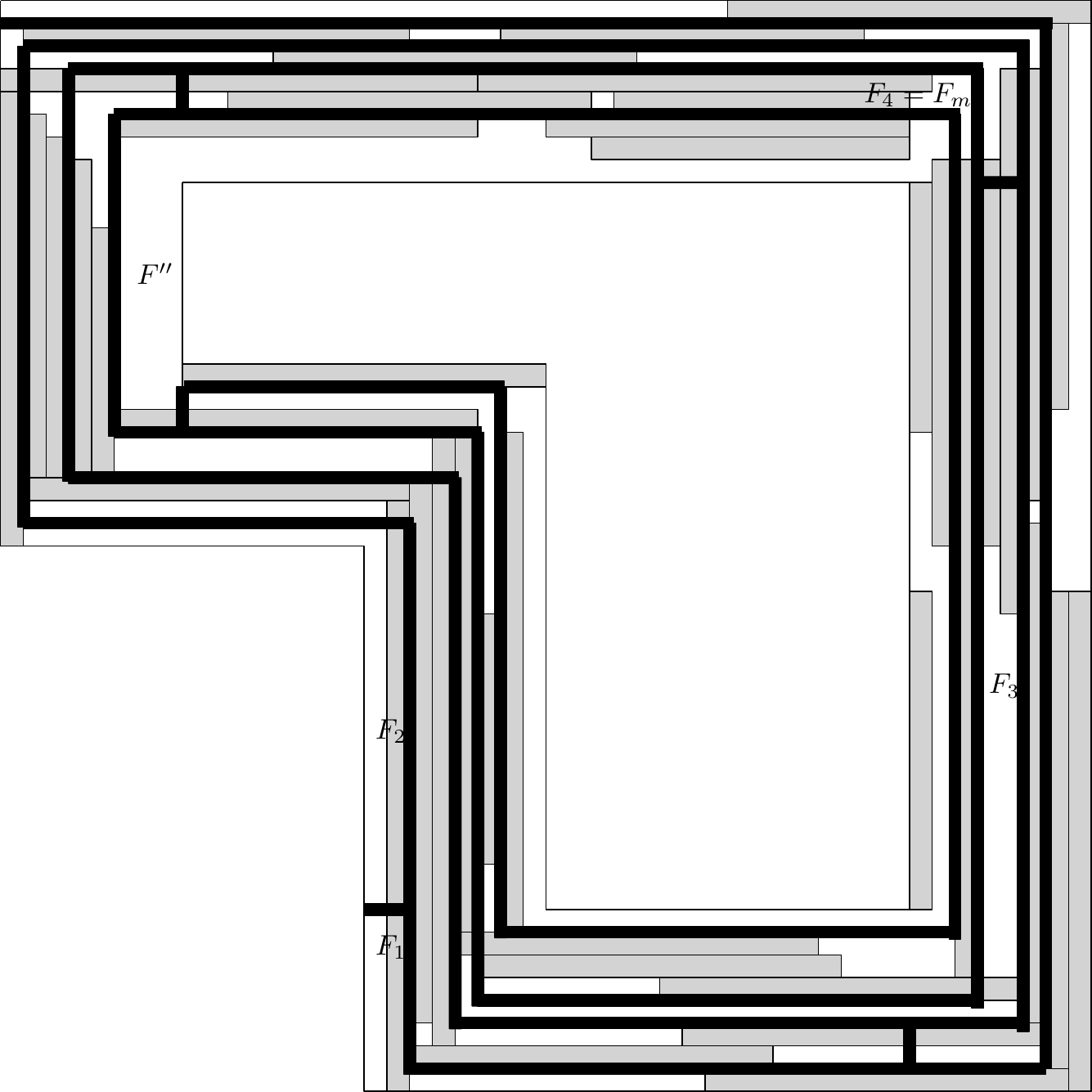}
}
\caption{\label{fig:slice-cycle-faces-many-paths}Partitioning a cycle-face into a set of paths $F_{1},...,F_{m}$ and a cycle face $F''$.}
\end{figure}

Now let us focus on $F'$. 
In case the boundary of $F'$ is more complex than our upper bound of
$\frac{1}{\eps}\cdot(\frac{1}{\delta})^{O(1)}$ edges allows,
we split $F'$ into a set of consecutive paths $F_{1},...,F_{m}$, each of them
having a boundary whose complexity is at most $\frac{1}{\eps}\cdot(\frac{1}{\delta})^{O(1)}$.
We perform these necessary cuts along grid lines. 
Each connected component of $F' \cap Q$ is a rectangle or an L-shape (where we treat a rectangular component $C$ within a cell $Q$ like an L-shape if $C\cap e \ne \emptyset \ne C \cap e'$ for two perpendicular
edges $e,e'$ of $Q$).
Let $C_1,\ldots,C_{m'}$ be the set of all consecutive connected components of $F'$ within single grid cells such that $C_i$ is an L-shape. For each $1/\eps$-th component $C_i$, starting with a random offset, we cut $F'$ inside $C_i$, along a grid cell boundary (intuitively, we cut $F'$ after every $1/\eps$ bends). There is an offset for which the total weight of intersected blocks is at most
a $\eps$-fraction of the total weight of the blocks in $F'$. 
Each resulting path $F_{i}$ can be expressed as the union of a set
of at most $O(1/(\eps \delta))$ components, where each component is a rectangle or an L-shape within some component $C$ of $F$, and therefore the complexity of $F_{i}$ is at most $\frac{1}{\eps}\cdot(\frac{1}{\delta})^{O(1)}$.

One can show, in the same way as when upper bounding the complexity of the boundary of $F''$, that for any $m'\le m$ the total boundary of each set
$F_{m'}\cup\partial F_{m'}\cup F_{m'+1}\cup\partial F_{m'+1}\cup \ldots \cup F_{m}\cup\partial F_{m} \cup F'' \cup \partial F''$
has at most $\frac{1}{\eps}\cdot(\frac{1}{\delta})^{O(1)}$
edges. Hence, GEO-DP parametrized with $k\ge\frac{1}{\eps}\cdot(\frac{1}{\delta})^{O(1)}$
tries to partition $F$ step by step, in the $i$-th step splitting off the path $F_{i}$ from the remaining area of $F$. Finally, it partitions $F$ into path-faces $F_{1},...,F_{m}$ and a cycle-face $F''$.
Knowing from the proof of Lemma~\ref{lem:dp-for-paths} that GEO-DP
solves the subproblem for each path optimally if $k$ is sufficiently
large compared to the length of the path, we conclude that GEO-DP parametrized with $k\ge \frac{1}{\eps} \cdot(\frac{1}{\delta})^{O(1)}$
computes a solution for $F'$ with weight at least $(1-\eps)w(F')$. By continuing
with the same arguments for $F''$, in case $F'' \neq \emptyset$, we get that GEO-DP computes a solution
for $F$ whose weight is at least $(1-\eps)w(F)$. 

With similar adjustments
as in the proof of Lemma~\ref{lem:dp-for-paths} we show that the
same is true in the setting of arbitrary rectangles at the cost of
an increase in complexity by a factor of $(\frac{1}{\delta})^{O(1)}$.\end{proof}

\begin{proof}[Proof of Theorem~\ref{thm:ptas-for-large}]
Fix $k:=(1/\eps)(1/\delta)^{O(1)}$.
For each set of rectangles $\R$ and for each polygon in the original input instance there is a corresponding polygon, containing
the same subset of rectangles,
in the instance obtained after the preprocessing performed by GEO-DP (which ensures that
all corners of rectangles have coordinates within $\{0,...,2n-1\}$, where $n=|\R|$).
Hence, it suffices to show that GEO-DP, when executing without the preprocessing routine, achieves the claimed approximation ratio 
on the original input instance.

Let $\R$ be an optimal solution for a given instance of the problem, and let $\L$ be the set of lines constructed for $\R$. From Lemma \ref{lem:maze-properties} we have $|\L|\le (1/\eps)(1/\delta)^{O(1)}$, and so $\L$ partitions the original input square into at most $k$ faces, and GEO-DP will consider such a partition of the input square.
From Lemma \ref{lem:maze-properties}, the total
weight of intersected rectangles from the optimal solution $\R$ in this partition is upper
bounded by $\eps\cdot w(\R)$. When recursing on each resulting face,
due to Lemmas~\ref{lem:dp-for-paths} and \ref{lem:dp-for-cycles}
GEO-DP obtains a $(1+\eps)$-approximative solution for each subproblem.
Thus, in total we obtain an approximation ratio of $1+O(\eps)$.

When GEO-DP performs the preprocessing first, in the modified input instance it will consider the same recursive partitioning of the input square, obtaining the same approximation ratio. From Proposition \ref{prop:GEO-DP-running-time}, the running time of GEO-DP parametrized with $k=(1/\eps)(1/\delta)^{O(1)}$ is upper bounded by $n^{(1/\eps)^2(1/\delta)^{O(1)}}$, and so it yields a PTAS for $\delta$-large rectangles for any constant $\delta > 0$.
\end{proof}

\begin{proof}[Proof of Corollary~\ref{cor:PTAS-roughly-same-size}] 
Suppose that we are given
an input  instance for which there is a value $K$ such that for any rectangle $R_i$ we have $\max\{g_{i},h_{i}\}\in[K,1/\delta\cdot K]$,
and all input coordinates are within $\{0,...,N\}$ for some integer
$N$. We define a grid with offset $a\in\mathbb{N}$ whose grid
cells have height and width $1/\delta\cdot K/\eps$, i.e., we define
a cell $[a+i\cdot1/\delta\cdot K/\eps,a+(i+1)\cdot1/\delta\cdot K/\eps]\times[a+j\cdot1/\delta\cdot K/\eps,a+(j+1)\cdot1/\delta\cdot K/\eps]$
for each $i,j\in\{-1,...,\left\lceil N\eps/(1/\delta \cdot K)\right\rceil \}$. 

Since $\max\{g_{i},h_{i}\}\le (1/\delta)\cdot K$ for each rectangle $R_{i}$,
when choosing a random offset $a\in\{0,...,\left\lceil 1/\delta\cdot K/\eps\right\rceil \}$,
in expectation the intersected rectangles from an optimal solution
$OPT$ have a total weight of at most $O(\eps)\cdot w(OPT)$. We take
this random offset and consider the resulting subproblems in each
grid cell which contains at least one input rectangle (there can be
at most $n$ such grid cells), i.e., each subproblem consists of the original instance
restricted to the rectangles contained in the respective grid cell. 
Since each grid cell has height and width $1/\delta\cdot K/\eps$
and $\max\{g_{i},h_{i}\}\in[K,1/\delta\cdot K]$ for each rectangle
$R_{i}$, we conclude that each grid cell constitutes a subinstance
in which all rectangles are $(\eps\cdot \delta)$-large. 
During the recursion process, GEO-DP guesses exactly this subdivision 
(not all at once but e.g., in a quad-tree fashion).
Using Theorem~\ref{sec:PTAS-large-rectangles} we know that
GEO-DP computes a $(1+\eps)$-approximation for each of them, given
that $k\ge (1/\eps)(\eps\cdot \delta)^{O(1)}$.
Altogether, we obtain
a $(1+O(\eps))$-approximation algorithm for the overall problem.
\end{proof}

\end{document}